\documentclass[a4paper,fleqn,11pt]{article}

\usepackage{amsmath}
\usepackage{amsthm}
\usepackage{amssymb}
\usepackage{accents}

\usepackage[a4paper,top=3cm, bottom=3cm, left=3cm, right=3cm]{geometry}
\usepackage[shortlabels,inline]{enumitem}
\usepackage[pdftex]{graphicx}
\usepackage{subcaption}
\usepackage[pdftex,ocgcolorlinks,urlcolor=blue,pagebackref=false]{hyperref}

\setlist[enumerate,1]{label={(\roman*)}}

\usepackage[capitalise,noabbrev]{cleveref}

\Crefname{property}{Property}{Properties}

\theoremstyle{plain}
\newtheorem{theorem}{Theorem}[section]
\newtheorem{lemma}[theorem]{Lemma}
\newtheorem{proposition}[theorem]{Proposition}
\newtheorem{corollary}[theorem]{Corollary}
\newtheorem{remark}[theorem]{Remark}

\theoremstyle{definition}
\newtheorem{definition}[theorem]{Definition}

\newcommand{\reals}{\mathbb{R}}
\newcommand{\complexes}{\mathbb{C}}
\newcommand{\naturals}{\mathbb{N}}
\newcommand{\positivereals}{\mathbb{R}_{>0}}
\newcommand{\positiveintegers}{\mathbb{N}_{>0}}
\newcommand{\nonnegativereals}{\mathbb{R}_{\ge 0}}

\newcommand{\distributions}[1][]{\mathcal{P}_{#1}}
\DeclareMathOperator{\probability}{Pr}

\newcommand{\entropy}{H}
\newcommand{\relativeentropy}[3][]{\mathop{D_{#1}}\mathopen{}\left(#2\middle\|#3\right)\mathclose{}}

\newcommand{\typeclass}[2]{T^{#1}_{#2}}
\newcommand{\typeclassprojector}[2]{\Pi^{#1}_{#2}}

\newcommand{\ket}[1]{\left|#1\right\rangle}
\newcommand{\bra}[1]{\left\langle #1\right|}
\newcommand{\ketbra}[2]{\left|#1\middle\rangle\!\middle\langle#2\right|}
\newcommand{\braket}[2]{\left\langle#1\middle|#2\right\rangle}

\newcommand{\loccto}[1][]{\xrightarrow{\textnormal{LOCC}}_{#1}}
\newcommand{\W}{\textnormal{W}}
\newcommand{\EPR}{\textnormal{EPR}}

\DeclareMathOperator{\support}{supp}

\DeclareMathOperator{\boundeds}{\mathcal{B}}
\DeclareMathOperator{\states}{\mathcal{S}}

\DeclareMathOperator{\domain}{dom}
\DeclareMathOperator{\epigraph}{epi}

\DeclareMathOperator{\convexhull}{conv}

\DeclareMathOperator{\Tr}{Tr}

\newcommand{\norm}[2][]{\left\|#2\right\|_{#1}}
\newcommand{\ball}[2]{B_{#1}(#2)}

\newcommand{\setbuild}[2]{\left\{#1\middle|#2\right\}}

\newcommand{\partitions}[1][]{\mathcal{P}_{#1}}

\newcommand{\ed}{\mathop{}\!\mathrm{d}}

\usepackage[affil-it]{authblk}

\usepackage{pgfplots}
\pgfplotsset{compat=1.16}
\usepgfplotslibrary{fillbetween}
\usetikzlibrary{patterns}
\pgfplotsset{
  rateplotstyle/.style={
    grid=both, 
    axis x line=bottom,
    axis y line=left,
    xmin=0,
    xmax=0.7,
    xlabel=$r$,
    xlabel style={anchor=north},
    ylabel=$R$,
    ylabel style={anchor=south east},
    width=0.6\textwidth, 
    height=0.45\textwidth, 
  },
  plotline/.style={
    black,
    semithick,
    mark=none,
  },
}

\title{Error exponents for tripartite-to-bipartite entanglement transformations}
\author{P\'eter Vrana}
\affil{Department of Algebra and Geometry, Institute of Mathematics, Budapest University of Technology and Economics, M\H uegyetem rkp. 3., H-1111 Budapest, Hungary.}

\begin{document}
\maketitle

\begin{abstract}
We consider distillation of ebits between a specified pair of subsystems from pure tripartite states by local operations and classical communication. It is known that, allowing an asymptotically vanishing error, the maximal rate is the minimum of the von Neumann entropies of the two corresponding marginals, and under asymptotic stochastic local operations and classical communication the maximal rate is given by a minimization over a one-parameter family of entanglement measures. In this paper, we determine the direct and strong converse error exponents, and the optimal rate for deterministic transformations.
\end{abstract}

\section{Introduction}

One of the main goals in entanglement theory is to characterize transformations between entangled states that are possible by local operations and classical communication. As a rule of thumb, this becomes difficult when more than two subsystems of mixed states are involved.\textsuperscript{[\href{https://en.wikipedia.org/wiki/Wikipedia:Citation_needed}{citation needed}]}

Transformations between pure bipartite states are fairly well understood in various settings. Single-shot deterministic transformability is characterized by the majorization between the vectors of Schmidt coefficients \cite{nielsen1999conditions}, and this can be extended to single-shot probabilistic transformations using weak majorization \cite{vidal1999entanglement}. The optimal protocol for probabilistically obtaining maximally entangled states was found in \cite{lo2001concentrating}. 

Asymptotic resuls for transforming many copies of a state into other states are either obtained by employing single-shot results and analyzing the relevant parameters of i.i.d. states, or independent methods. The central result on transformations between bipartite pure states with asymptotically vanishing error is that the optimal rate is the ratio of the (von Neumann) entropies of entanglement of the two states \cite{bennett1996concentrating}. This has been refined in the case of entanglement concentration, i.e., transforming i.i.d. copies of an arbitrary pure state into maximally entangled states in various ways.

When only deterministic transformations are allowed (probability $1$ for every finite number of copies), then the optimal rate is the min-entropy of entanglement \cite{morikoshi2001deterministic}, which is can be much smaller than the von Neumann entropy of entanglement. At the other end of the scale, when we are satisfied with any nonzero probability, it is an exercise in linear algebra to see that the optimal rate is the R\'enyi entanglement entropy of order $0$, i.e., the logarithm of the number of nonzero Schmidt coefficients. 

Between these two extreme cases, the success probability approaches either $0$ or $1$ exponentially fast with an exponent $r$ depending on the rate. The trade-off between the two has been determined by Hayashi et al.\ in \cite{hayashi2002error} in both cases, and they obtained similar results when approximate transformations are considered with the fidelity replacing the success probability. For probabilistic exact transformations with exponentially decreasing success probability, the optimal rate is known even for arbitrary initial and target states \cite{jensen2019asymptotic}.

For multipartite states, the situation is much more difficult, largely because of the lack of an analogue of the Schmidt decomposition for arbitrary multipartite states. This means that even the structure of different kinds of multipartite entanglement are little understood, which makes it very difficult to formulate and prove strong results on entanglement transformations that hold for all pairs of states.

In this paper we study a restricted class of transformations. The initial state can be any tripartite pure state $\psi$, but the desired target is a maximally entangled pair between the $A$ and $B$ subsystems, which is one of the possible extensions of entanglement concentration. Such transformatinos have been considered in previous work. In the one-shot setting this is known as the entanglement of assistance \cite{cohen1998unlocking,divincenzo1999entanglement,laustsen2003local}. In the SLOCC paradigm, such transformations are connected to matrix subspaces \cite{li2018tripartite}. In the asymptotic setting with asymptotically vanishing error, the optimal rate (measured by the number of ebits obtained per copy) is equal to the minimum of the two marginal von Neumann entropies $\entropy(A)_\psi$ and $\entropy(B)_\psi$ \cite{smolin2005entanglement}. When the transformation is required to be exact, but only succesful with an arbitrarily small nonzero probability (asymptotic SLOCC), the optimal rate has recently been determined \cite{christandl2023weighted}, and can be expressed as a minimization over a one-parameter family of entanglement measures, a subset of the quantum functionals introduced in \cite{christandl2023universal}.

The quantum functionals have been extended to a larger family by including a parameter that is similar to the order parameter in R\'enyi information quantities \cite{vrana2023family}. Such entanglement measures are known to play a role in characterizing the trade-off between the rate and the strong converse exponent for general transformations \cite{jensen2019asymptotic}. These results may give some hope that the error exponents for tripartite-to-bipartite transformations can be determined as well, which is the goal of the present paper.

\subsection{Results}

We show that the close connection between the fidelity and probabilty-based notions of asymptotic transformations that was found in the case of entanglement concentration remains true for transformations from arbitrary pure tripartite states to EPR pairs between specified subsystems. In the direct domain (exponentially decaying error), this simply means (see \cref{sec:preliminaries} for precise definitions of the rates and the appearing entropic quantities)
\begin{equation}
R_F(\psi\to\EPR_{AB},r)=R(\psi\to\EPR_{AB},r),
\end{equation}
while for the rates as functions of the strong converse exponent we have
\begin{equation}
R^*_F(\psi\to\EPR_{AB},r)=\sup_{0\le x\le r}\left[R^*(\psi\to\EPR_{AB},r)+r-x\right]
\end{equation}
(\cref{prop:directexponentsequal,prop:strongconverseexponentrelation}). For this reason, our main focus will be exact probabilistic transformations.

For large success probability or fidelity, we find that the upper bounds obtained by considering bipartitions are tight. Specifically, we show that a tripartite state $\ket{\psi}\in\mathcal{H}_A\otimes\mathcal{H}_B\otimes\mathcal{H}_C$ can be transformed into
\begin{equation}
\left\lceil \min\{\entropy_\infty(A)_\psi,\entropy_\infty(B)_\psi\}-\log\left(4\ln(2d_C^4(d_A+d_B))\right)\right\rceil
\end{equation}
EPR pairs between the $A$ and $B$ subsystems (\cref{thm:oneshotdeterministic}), while it follows from Nielsen's theorem applied to the bipartitions $A:BC$ and $AC:B$ that it is not possible to obtain more than $\left\lfloor \min\{\entropy_\infty(A)_\psi,\entropy_\infty(B)_\psi\}\right\rfloor$. Asymptotically, this implies that the optimal deterministic rate is equal to $\min\{\entropy_\infty(A)_\psi,\entropy_\infty(B)_\psi\}$ (\cref{cor:deterministicrate}).

Relaxing the requirement to a lower bound of $1-2^{-rn+o(n)}$ on either the success probability of an exact transformation or the fidelity of an approximate transformation results in larger rates. The bipartitions give the upper bound $R(\psi\to\EPR_{AB},r)\le\min\{E_A(r),E_B(r)\}$ with
\begin{equation}
E_A(r)=\sup_{\alpha>1}\left[r\frac{1}{1-\alpha}+\entropy_\alpha(A)_\psi\right]
\end{equation}
and analogously for $B$. We show that this is in fact an equality (\cref{thm:directexponent}), by adapting the bipartite techniques and combining them with our new one-shot lower bound for tripartite-to-bipartite transformations.

In contrast, for exponentially decaying success probability or fidelity, the minimum of the two corresponding bipartite upper bounds is not achievable, and to express the rate we need the entanglement measures introduced in \cite{vrana2023family} and further investigated in \cite{bugar2024interpolating,bugar2024explicit}. These are given in terms of representation-theoretic data defined via the local Schur--Weyl projections, and were previously known to provide upper bounds on the asymptotic rate under exponentially decaying success probability. Here we show that the minimum of these upper bounds is achievable and one can even restrict to a two-parameter subfamily in the minimization (\cref{thm:generalscrate}):
\begin{equation}
R^*(\psi\to\EPR_{AB},r)=\inf_{\substack{\alpha\in[0,1)  \\  x\in[0,1]}}\left[r\frac{\alpha}{1-\alpha}+E^{\alpha,(x,1-x,0)}(\psi)\right].
\end{equation}
For a special subclass of states that admit a type of simultaneous Schmidt decomposition, considered in \cite{bugar2024explicit}, an alternative expression is available. Denoting by $P$ the measured distribution in the preferred product basis (see \cref{sec:strongconverse} for details), the expression is
\begin{equation}
R^*(\psi\to\EPR_{AB},r)=\max_{\substack{Q\in\distributions(\support P)  \\  \relativeentropy{Q}{P}\le r}}\min\{\entropy(Q_A),\entropy(Q_B)\}
\end{equation}
(see \cref{thm:freesupportscrate}).

\subsection{Overview}

In \cref{sec:preliminaries} we introduce the precise setup and review some of the necessary background. This includes the comparison of fidelity-based and probability-based asymptotic error criteria, which are related by an inequality in general, but in our case (like for entanglement concentration), there is a precise quantitative relation. The main ideas here are not new, but extend those in \cite{hayashi2002error}, where only bipartite entanglement concentration is considered.

In \cref{sec:deterministic} we prove our one-shot lower bound and as a corollary we obtain the optimal rate for deterministic transformations. In \cref{sec:direct} we address the convergence of the success probability to $1$ using a truncation argument to reduce to the deterministic case. In \cref{sec:strongconverse} we consider exponentially decaying success probability (strong converse). A crucial step here is the derivation of a new expression for the entanglement measures $E^{\alpha,\theta}$, which is similar but better-behaved than the original formula in \cite{vrana2023family}. Since the proofs of \cref{thm:freesupportscrate,thm:generalscrate} are independent and the special case is similar but does not require the new expression, whe chose to present it first separately. In \cref{sec:conclusion} we outline the missing ingredients for a possible extension of the results to multipartite states.

\section{Preliminaries}\label{sec:preliminaries}

\subsection{Notations}

In this paper we consider tripartite pure states. The subsystems will be labelled by $A,B,C$, and the subsystem Hilbert spaces are $\mathcal{H}_A,\mathcal{H}_B,\mathcal{H}_C$, which are always assumed to be of finite dimensions $d_A,d_B,d_C$, respectively. An ebit between subsystems $A$ and $B$ (and in a product state with $C$) will be denoted by $\EPR_{AB}$.

When $\rho$ and $\sigma$ are positive operators on (possibly different) tripartite Hilbert spaces, the notation $\rho\loccto\sigma$ means that there is a trace-nonincreasing LOCC channel $\Lambda$ such that $\Lambda(\rho)=\sigma$ (typically $\rho$ is a state and $\sigma$ has trace at most $1$). Approximate transformations are understood with respect to the fidelity $F(\sigma,\tau)=\norm[1]{\sqrt{\sigma}\sqrt{\tau}}^2$. When $\tau=\ketbra{\varphi}{\varphi}$ is a pure state, this can be simplified to $F(\sigma,\ketbra{\varphi}{\varphi})=\bra{\varphi}\sigma\ket{\varphi}$. $\rho\loccto[\epsilon]\sigma$ means that there is a trace-preserving LOCC channel $\Lambda$ such that $F(\Lambda(\rho),\sigma)\ge 1-\epsilon$.

The set of probability distributions on a finite set $\mathcal{X}$ will be denoted by $\distributions(\mathcal{X})$. We may identify these with maps $P:\mathcal{X}\to\nonnegativereals$ such that $\sum_{x\in\mathcal{X}}P(x)=1$. The support of $P\in\distributions(\mathcal{X})$ is $\support P=\setbuild{x\in\mathcal{X}}{P(x)\neq 0}$. The R\'enyi entropy of order $\alpha\in(0,1)\cup(1,\infty)$ is (using base-$2$ logarithm)
\begin{align}
\entropy_\alpha(P) & = \frac{1}{1-\alpha}\log\sum_{x\in\mathcal{X}}P(x)^\alpha,
\intertext{which is extended continuously to $\alpha\in[0,\infty]$ as}
\entropy_0(P) & = \log\lvert\support P\rvert  \\
\entropy_1(P) & = -\sum_{x\in\support P}P(x)\log P(x)  \\
\entropy_\infty(P) & = -\log\max_{x\in\mathcal{X}}P(x).
\end{align}

For $P,Q\in\distributions(\mathcal{X})$, the relative entropy (or Kullback--Leibler divergence) is
\begin{equation}
\relativeentropy{P}{Q}=\begin{cases}
\sum_{x\in\support P}P(x)\log\frac{P(x)}{Q(x)} & \text{if $\support P\subseteq\support Q$}  \\
\infty & \text{otherwise.}
\end{cases}
\end{equation}

\subsection{Entanglement transformations}

Asymptotic entanglement transformations can be studied under a variety of error criteria. In each case we aim to find the largest rate compatible with a given lower bound on the success probability or fidelity,
\begin{align}
R_{\textnormal{det}}(\rho\to\sigma) & = \sup\setbuild{\limsup_{m\to\infty}\frac{n_m}{m}}{\rho^{\otimes m}\loccto \sigma^{\otimes n_m}}  \\
R(\rho\to\sigma) & = \sup\setbuild{\limsup_{m\to\infty}\frac{n_m}{m}}{\lim_{m\to\infty}p_m=1\text{ and }\rho^{\otimes m}\loccto p_m\sigma^{\otimes n_m}}  \\
R_F(\rho\to\sigma) & = \sup\setbuild{\limsup_{m\to\infty}\frac{n_m}{m}}{\lim_{m\to\infty}\epsilon_m=0\text{ and }\rho^{\otimes m}\loccto[\epsilon_m]\sigma^{\otimes n_m}}  \\
R(\rho\to\sigma,r) & = \sup\setbuild{\limsup_{m\to\infty}\frac{n_m}{m}}{\liminf_{m\to\infty}-\frac{1}{m}\log(1-p_m)\ge r\text{ and }\rho^{\otimes m}\loccto p_m\sigma^{\otimes n_m}}  \\
R_F(\rho\to\sigma,r) & = \sup\setbuild{\limsup_{m\to\infty}\frac{n_m}{m}}{\liminf_{m\to\infty}-\frac{1}{m}\log \epsilon_m\ge r\text{ and }\rho^{\otimes m}\loccto[\epsilon_m]\sigma^{\otimes n_m}}  \\
R^*(\rho\to\sigma,r) & = \sup\setbuild{\limsup_{m\to\infty}\frac{n_m}{m}}{\limsup_{m\to\infty}-\frac{1}{m}\log p_m\le r\text{ and }\rho^{\otimes m}\loccto p_m\sigma^{\otimes n_m}}  \\
R^*_F(\rho\to\sigma,r) & = \sup\setbuild{\limsup_{m\to\infty}\frac{n_m}{m}}{\limsup_{m\to\infty}-\frac{1}{m}\log(1-\epsilon_m)\le r\text{ and }\rho^{\otimes m}\loccto[\epsilon_m]\sigma^{\otimes n_m}}  \\
R^*(\rho\to\sigma,\infty) & = \sup\setbuild{\limsup_{m\to\infty}\frac{n_m}{m}}{\forall m:p_m>0\text{ and }\rho^{\otimes m}\loccto p_m\sigma^{\otimes n_m}}
\end{align}
where in the conditions it is understood that $(n_m)_{m\in\naturals}\in\naturals^\naturals$ and there exists $(p_m)_{m\in\naturals}\in(0,1)^\naturals$ (respectively $(\epsilon_m)_{m\in\naturals}\in(0,1)^\naturals$) with the stated properties and the transformation is possible for all $m$. Note that $R^*(\rho\to\sigma,\infty)$ does not have a counterpart for approximate transformations, since even a separable state can have nonzero fidelity to an arbitrary entangled one.

It is clear from the definitions that
\begin{equation}
R_{\textnormal{det}}(\rho\to\sigma)\le R(\rho\to\sigma,r)\le R(\rho\to\sigma)\le R^*(\rho\to\sigma,r')
\end{equation}
and
\begin{equation}
R_{\textnormal{det}}(\rho\to\sigma)\le R_F(\rho\to\sigma,r)\le R_F(\rho\to\sigma)\le R^*_F(\rho\to\sigma,r')
\end{equation}
and that $R(\rho\to\sigma,r)$ and $R_F(\rho\to\sigma,r)$ are decreasing functions of $r$ while $R^*(\rho\to\sigma,r)$ and $R^*_F(\rho\to\sigma,r)$ are increasing functions. In addition, since $\rho\loccto (1-\epsilon)\sigma$ implies $\rho\loccto[\epsilon]\sigma$, the inequalities
\begin{align}
R(\rho\to\sigma,r) & \le R_F(\rho\to\sigma,r)  \label{eq:directprobabilityvsfidelity}  \\
R(\rho\to\sigma) & \le R_F(\rho\to\sigma)  \\
R^*(\rho\to\sigma,r) & \le R^*_F(\rho\to\sigma,r)
\end{align}
hold.

The tensor product of LOCC channels is also an LOCC channel. Together with the multiplicativity of the success probability (the probability for independently run protocols to be both successful), and the fidelity, this implies that if $\rho_1\loccto p_1\sigma_1$ and $\rho_2\loccto p_2\sigma_2$, then $\rho_1\otimes\rho_2\loccto p_1p_2\sigma_1\otimes\sigma_2$, and if $\rho_1\loccto[\epsilon_1]\sigma_1$ and $\rho_2\loccto[\epsilon_2]\sigma_2$, then $\rho_1\otimes\rho_2\loccto[(1-\epsilon_1)(1-\epsilon_2)]\sigma_1\otimes\sigma_2$.
\begin{proposition}\label{prop:suprates}
\begin{align}
R_{\textnormal{det}}(\rho\to\sigma) & = \sup\setbuild{\frac{n}{m}}{\rho^{\otimes m}\loccto\sigma^{\otimes n}}  \\
R^*(\rho\to\sigma,r) & = \sup\setbuild{\frac{n}{m}}{\rho^{\otimes m}\loccto 2^{-rm}\sigma^{\otimes n}}  \\
R^*(\rho\to\sigma,\infty) & = \sup\setbuild{\frac{n}{m}}{\exists p>0:\rho^{\otimes m}\loccto p\sigma^{\otimes n}}  \\
R^*_F(\rho\to\sigma,r) & = \sup\setbuild{\frac{n}{m}}{\rho^{\otimes m}\loccto[1-2^{-rm}]\sigma^{\otimes n}}
\end{align}
In addition, $r\mapsto R^*(\rho\to\sigma,r)$ and $r\mapsto R^*_F(\rho\to\sigma,r)$ are concave functions.
\end{proposition}
\begin{proof}
We prove the first equality, the remaining three can be seen in a similar way. If $\sigma$ is separable, then both sides are $\infty$, therefore we may assume that $\sigma$ is entangled. Let $a_m=\max\setbuild{n\in\naturals}{\rho^{\otimes m}\loccto\sigma^{\otimes n}}$. Then clearly
\begin{equation}
R_{\textnormal{det}}(\rho\to\sigma)
 = \limsup_{m\to\infty}\frac{a_m}{m}.
\end{equation}
For all $m_1$ and $m_2$, $\rho^{\otimes m_1}\loccto\sigma^{\otimes a_{m_1}}$ and $\rho^{\otimes m_2}\loccto\sigma^{\otimes a_{m_2}}$, therefore $\rho^{\otimes (m_1+m_2)}\loccto\sigma^{\otimes(a_{m_1}+a_{m_2})}$, which implies $a_{m_1+m_2}\ge a_{m_1}+a_{m_2}$. Clearly also $a_m\ge 0$ for all $m$. By the Fekete lemma, we have
\begin{equation}
\limsup_{m\to\infty}\frac{a_m}{m}=\lim_{m\to\infty}\frac{a_m}{m}=\sup_{m\in\positiveintegers}\frac{a_m}{m}.
\end{equation}

Concavity follows by a time-sharing argument. Let $r_1,r_2>0$ and $\lambda\in(0,1)$. By the previous point, there are sequences $n_{1,m}$, $n_{2,m}$, $p_{1,m}$, and $p_{2,m}$ such that $\rho^{\otimes m}\loccto p_{1,m}\sigma^{\otimes n_{1,m}}$ and $\rho^{\otimes m}\loccto p_{2,m}\sigma^{\otimes n_{2,m}}$ hold for all $m$, $\limsup_{m\to\infty}-\frac{1}{m}\log p_{i,m}\le r_i$ 
\end{proof}

The precise values of the optimal rates are unknown even for transformations between general bipartite pure states. The only exception is when $\rho$ is a pure bipartite state and $\sigma$ is maximally entangled, without loss of generality an ebit $\EPR=(\ket{00}+\ket{11})/\sqrt{2}$. Up to local unitary transformations, a bipartite pure state can be written as
\begin{equation}
\ket{\psi_P}=\sum_{x\in\mathcal{X}}\sqrt{P(x)}\ket{x}_A\otimes\ket{x}_B,
\end{equation}
where $P\in\distributions(\mathcal{X})$ is a probability distribution on a finite set $\mathcal{X}$. The rates are
\begin{align}
R_{\textnormal{det}}(\psi_P\to\EPR) & = \entropy_\infty(P)
 & \text{\cite{morikoshi2001deterministic}}  \\
R(\psi_P\to\EPR,r) & = \sup_{\alpha>1}\left[r\frac{1}{1-\alpha}+\entropy_\alpha(P)\right]
 & \text{\cite{hayashi2002error}}  \\
R(\psi_P\to\EPR) & = \entropy(P)
 & \text{\cite{bennett1996concentrating}}  \\
R^*(\psi_P\to\EPR,r) & = \inf_{\alpha\in[0,1)}\left[r\frac{\alpha}{1-\alpha}+\entropy_\alpha(P)\right]
 & \text{\cite{hayashi2002error}}  \\
R^*(\psi_P\to\EPR,\infty) & = \entropy_0(P) 
\intertext{and}
R_F(\psi_P\to\EPR,r) & = \sup_{\alpha>1}\left[r\frac{1}{1-\alpha}+\entropy_\alpha(P)\right]
 & \text{\cite{hayashi2002error}}  \\
R_F(\psi_P\to\EPR) & = \entropy(P)
 & \text{\cite{bennett1996concentrating}}  \\
R_F^*(\psi_P\to\EPR,r) & = \begin{cases}
\inf_{\alpha\in[0,1)}\left[r\frac{\alpha}{1-\alpha}+\entropy_\alpha(P)\right] & \makebox[0pt][l]{if $0< r\le r_0$}  \\
r-r_0+R_F^*(\psi_P\to\EPR,r_0) & $\makebox[0pt][l]{if $r_0<r$,}$
\end{cases}
\end{align}
where $r_0$ is characterized by $\left.\frac{\ed}{\ed r}R^*(\psi_P\to\EPR,r)\right|_{r=r_0}=1$ \cite{hayashi2002error}.

It is apparent in the above formulas that the fidelity-based and the probabilistic optimal rates are closely related in this case, which is however not true for general transformations. In the case or pure bipartite states (arbitrary initial and target state), the following rates are known:
\begin{align}
R_{\textnormal{det}}(\psi_P\to\psi_Q) & = \min_{\alpha\in[0,\infty]}\frac{\entropy_\alpha(P)}{\entropy_\alpha(Q)} & \text{\cite{jensen2019asymptoticmajorization}}  \\
R(\psi_P\to\psi_Q) & = \inf_{\alpha\in[0,1)}\frac{\entropy_\alpha(P)}{\entropy_\alpha(Q)} & \text{\cite{jensen2019asymptotic}}  \\
R^*(\psi_P\to\psi_Q,r) & = \inf_{\alpha\in[0,1)}\frac{r\frac{\alpha}{1-\alpha}+\entropy_\alpha(P)}{\entropy_\alpha(Q)} & \text{\cite{jensen2019asymptotic}}  \\
R^*(\psi_P\to\psi_Q,\infty) & = \frac{\entropy_0(P)}{\entropy_0(Q)}  \\
R_F(\psi_P\to\psi_Q) & = \frac{\entropy(P)}{\entropy(Q)} & \text{\cite{bennett1996concentrating}}.
\end{align}
In particular, $R(\psi_P\to\psi_Q)$ can be strictly less than $R_F(\psi_P\to\psi_Q)$ in general.

In the case of entanglement concentration, the equality $R(\psi_P\to\EPR,r)=R_F(\psi_P\to\EPR,r)$ and the precise relation between $R^*(\psi_P\to\EPR,r)$ and $R^*_F(\psi_P\to\EPR,r)$ follows from Lemmas 8, 9, 12, and 13 in \cite{hayashi2002error} (the simplest one, \cite[Lemma 8]{hayashi2002error} holds much more generally, in the asymptotic setting its gives the inequality \eqref{eq:directprobabilityvsfidelity}). We now extend these lemmas to transformations from tripartite pure states to maximally entangled states between a specified pair
\begin{equation}
\ket{\phi_d}=\frac{1}{\sqrt{d}}\sum_{i=1}^d\ket{ii0}\in\complexes^d\otimes\complexes^d\otimes\complexes.
\end{equation}
\begin{lemma}\label{lem:highfidelitytoprobability}
Let $\ket{\psi}\in\mathcal{H}_A\otimes\mathcal{H}_B\otimes\mathcal{H}_C$ be a tripartite state vector, and suppose that $\ketbra{\psi}{\psi}\loccto[\epsilon]\ketbra{\phi_d}{\phi_d}$. Then with $L=\lfloor\frac{d}{8}\rfloor$ and $p=1-30\epsilon$ the relation $\ketbra{\psi}{\psi}\loccto p\ketbra{\phi_L}{\phi_L}$ holds.
\end{lemma}
\begin{proof}
Let $\Lambda$ be an LOCC channel such that $\rho:=\Lambda(\ketbra{\psi}{\psi})$ satisfies $\bra{\phi}\rho\ket{\phi}\ge 1-\epsilon$. $\Lambda$ can be written as a composition of measurements, classically controlled isometries and partial traces over the measurement results. By keeping track of all the measurement results, the protocol may be converted into one that transforms $\ket{\psi}$ into an ensemble of pure states $(p_i,\ket{\varphi_i})_{i\in I}$ such that $\sum_{i\in I}p_i\ketbra{\varphi_i}{\varphi_i}=\rho$. As only the first two subsystems of $\phi_d$ are nontrivial, each $\varphi_i$ is a bipartite pure state between $A$ and $B$, tensored with a fixed state on $C$. Let $\epsilon_i=1-\lvert\braket{\phi_d}{\varphi_i}\rvert^2$ and $L_i=\lfloor d(1-6\epsilon_i)/6\rfloor$. By \cite[Lemma 9]{hayashi2002error}, it is possible to transform $\ketbra{\varphi_i}{\varphi_i}$ into $\ket{\phi_{L_i}}$ with probability $1-6\epsilon_i$.

Since $\ketbra{\phi_d}{\phi_d}\loccto\ketbra{\phi_{d'}}{\phi_{d'}}$ whenever $d\ge g'$, it is therefore possible to transform $\ket{\psi}$ into $\ket{\phi_L}$ with probability at least
\begin{equation}
\sum_{\substack{i\in I  \\  L\le L_i}}p_i(1-6\epsilon_i)\ge \sum_{\substack{i\in I  \\  L\le L_i}}p_i-6\sum_{i\in I}p_i\epsilon_i.
\end{equation}
Since
\begin{equation}
1-\epsilon
 \le \bra{\phi}\rho\ket{\phi}
 = \sum_{i\in I}p_i\lvert\braket{\phi_d}{\varphi_i}\rvert^2
 = \sum_{i\in I}p_i(1-\epsilon_i),
\end{equation}
we have
\begin{equation}
\sum_{i\in I}p_i\epsilon_i\le\epsilon.
\end{equation}
The condition $L\le L_i=\lfloor d(1-6\epsilon_i)/6\rfloor$ is equivalent to $\epsilon_i\le \frac{1}{6}-\frac{L}{d}$ and, by the Markov inequality,
\begin{equation}
\sum_{\substack{i\in I  \\  \epsilon_i>\frac{1}{6}-\frac{L}{d}}}p_i\le\frac{\epsilon}{\frac{1}{6}-\frac{L}{d}}.
\end{equation}
It follows that the success probability is at least
\begin{equation}
1-\frac{\epsilon}{\frac{1}{6}-\frac{L}{d}}-6\epsilon
 = 1-12\frac{1-3\frac{L}{d}}{1-6\frac{L}{d}}\epsilon
 = 1-30\epsilon.
\end{equation}
\end{proof}
\begin{proposition}\label{prop:directexponentsequal}
For every tripartite pure state $\ket{\psi}\in\mathcal{H}_A\otimes\mathcal{H}_B\otimes\mathcal{H}_C$ and all $r>0$ we have $R(\psi\to\EPR_{AB},r)=R_F(\psi\to\EPR_{AB},r)$.
\end{proposition}
\begin{proof}
As noted above, the inequality $R(\psi\to\EPR_{AB},r)\le R_F(\psi\to\EPR_{AB},r)$ holds more generally (see \eqref{eq:directprobabilityvsfidelity}).

For the other direction, suppose that $\ketbra{\psi}{\psi}^{\otimes m}\loccto[\epsilon_m]\EPR_{AB}^{\otimes n_m}$ for some sequences $(\epsilon_m)_{m\in\naturals}$ and $(n_m)_{m\in\naturals}$. By \cref{lem:highfidelitytoprobability}, $\ketbra{\psi}{\psi}^{\otimes m}\loccto p_m\EPR_{AB}^{\otimes(n_m-3)}$ also holds with $p_m=1-30\epsilon_m$. Since
\begin{equation}
\liminf_{m\to\infty}-\frac{1}{m}\log(1-p_m)
 = \liminf_{m\to\infty}-\frac{1}{m}\log(30\epsilon_m)
 = \liminf_{m\to\infty}-\frac{1}{m}\log \epsilon_m
\end{equation}
and
\begin{equation}
\limsup_{m\to\infty}\frac{n_m-3}{m}
 = \limsup_{m\to\infty}\frac{n_m}{m},
\end{equation}
every feasible point in the definition of $R_F(\psi\to\EPR_{AB},r)$ is also a feasible point for $R(\psi\to\EPR_{AB},r)$, which implies $R(\psi\to\EPR_{AB},r)\ge R_F(\psi\to\EPR_{AB},r)$.
\end{proof}

\begin{lemma}\label{lem:lowprobabilitytofidelity}
Let $\ket{\psi}\in\mathcal{H}_A\otimes\mathcal{H}_B\otimes\mathcal{H}_C$ be a tripartite state vector, and suppose that $\ketbra{\psi}{\psi}\loccto p\ketbra{\phi_d}{\phi_d}$. Then for all $L\ge d$ the relation $\ketbra{\psi}{\psi}\loccto[1-p\frac{d}{L}]\ketbra{\phi_L}{\phi_L}$ holds.
\end{lemma}
\begin{proof}
The proof of \cite[Lemma 12]{hayashi2002error} extends in a straightforward way. There exists a state $\rho$ such that $\ketbra{\psi}{\psi}\loccto p\ketbra{\phi_d}{\phi_d}+(1-p)\rho$ (e.g., true for any separable $\rho$), and
\begin{equation}
F(p\ketbra{\phi_d}{\phi_d}+(1-p)\rho,\ketbra{\phi_L}{\phi_L})
 = \bra{\phi_L}(p\ketbra{\phi_d}{\phi_d}+(1-p)\rho)\ket{\phi_L}
 \ge p\lvert\braket{\phi_L}{\phi_d}\rvert
 = p\frac{d}{L}.
\end{equation}
\end{proof}

Next we prove a variant of \cite[Lemma 13]{hayashi2002error} that avoids the assumption that the closest state to $\ket{\phi_d}$ that is reachable from $\ket{\psi}$ by LOCC is a pure state, which is only known to hold for bipartite states. We use the following inequality from their proof.
\begin{lemma}\label{lem:sumlogestimate}
If $a_1,a_2,\dots,a_d\in\reals$, then
\begin{equation}
\sum_{j=1}^da_j\le a_1+(\ln d)\max_{j\in[d]}ja_j.
\end{equation}
\end{lemma}

\begin{lemma}\label{lem:lowfidelitytoprobability}
Let $\ket{\psi}\in\mathcal{H}_A\otimes\mathcal{H}_B\otimes\mathcal{H}_C$ be a tripartite state vector, and suppose that $\ketbra{\psi}{\psi}\loccto[\epsilon]\ketbra{\phi_d}{\phi_d}$.
Then there exists an integer $L\le d$ and $P\in(0,1]$ such that
\begin{equation}
\sqrt{PL}\ge\frac{\sqrt{d(1-\epsilon)}-1}{\ln d}
\end{equation}
and the relation $\ketbra{\psi}{\psi}\loccto P\ketbra{\phi_L}{\phi_L}$ holds.
\end{lemma}
\begin{proof}
As in the proof of \cref{lem:highfidelitytoprobability}, $\ket{\psi}$ can be transformed into an ensemble of pure biseparable (product with $C$) states $(p_i,\ket{\varphi_i})_{i\in I}$ such that $\sum_{i\in I}p_i(1-\epsilon_i)\ge 1-\epsilon$, where $\epsilon_i=1-\lvert\braket{\phi_d}{\varphi_i}\rvert^2$. Let $F_i=1-\epsilon_i$ and $F=\sum_{i\in I}p_iF_i$. We may assume that there is an ordered basis in which the states $\varphi_i$ and $\phi_d$ are all in Schmidt form with decreasing coefficients, as aligning the basis and sorting is possible with local unitary transformations, and would only increase the fidelities. Let
\begin{equation}
\ket{\varphi_i}=\sum_{j=1}^m\sqrt{Q^{(i)}(j)}\ket{j}\otimes\ket{j}
\end{equation}
for some $m\ge d$ (padding with zeros if necessary). Then
\begin{equation}
\sqrt{F_i}=\sum_{j=1}^d\frac{1}{\sqrt{d}}\sqrt{Q^{(i)}(j)},
\end{equation}
therefore
\begin{equation}
\begin{split}
\sqrt{dF}
 & = \frac{\sqrt{d}}{\sqrt{F}}F  \\
 & = \frac{\sqrt{d}}{\sqrt{F}}\sum_{i\in I}p_iF_i  \\
 & = \frac{\sqrt{d}}{\sqrt{F}}\sum_{i\in I}p_i\sqrt{F_i}\sum_{j=1}^d\frac{1}{\sqrt{d}}\sqrt{Q^{(i)}(j)}  \\
 & = \sum_{j=1}^d\sum_{i\in I}p_i\sqrt{\frac{F_i}{F}}\sqrt{Q^{(i)}(j)}  \\
 & \le 1+(\ln d)\max_{j\in[d]}j\sum_{i\in I}p_i\sqrt{\frac{F_i}{F}}\sqrt{Q^{(i)}(j)},
\end{split}
\end{equation}
in the last step using \cref{lem:sumlogestimate} and that the terms are bounded by $1$.

Let $j$ be an index where the maximum is attained. For the state $\ket{\varphi_i}$, we apply truncation at $Q^{(i)}(j)$ (see \cite[Section 2]{hayashi2002error} or the proof of \cref{thm:directexponent} below for more details on this step), after which it is possible to distill $\ket{\phi_L}$ with $L=j$ by Nielsen's theorem, with an overall probability $P_i\ge jQ^{(i)}(j)$.

Averaging over the ensemble (the outcomes of the approximate transformation), the total probability for transforming $\ket{\psi}$ into $\ket{\phi_L}$ satisfies
\begin{equation}
P=\sum_{i\in I}p_iP_i\ge j\sum_{i\in I}p_iQ^{(i)}(j)
\end{equation}

By the choice of $j$ and using the Cauchy--Schwarz inequality,
\begin{equation}
\begin{split}
\frac{\sqrt{dF}-1}{\ln d}
 & \le j\sum_{i\in I}p_i\sqrt{\frac{F_i}{F}}\sqrt{Q^{(i)}(j)}  \\
 & \le j\sqrt{\sum_{i\in I}p_i\frac{F_i}{F}}\sqrt{\sum_{i\in I}p_iQ^{(i)}(j)}  \\
 & = j\sqrt{\sum_{i\in I}p_iQ^{(i)}(j)}  \\
 & \le \sqrt{PL}.
\end{split}
\end{equation}
\end{proof}

\begin{proposition}\label{prop:strongconverseexponentrelation}
For every tripartite pure state $\ket{\psi}\in\mathcal{H}_A\otimes\mathcal{H}_B\otimes\mathcal{H}_C$ and all $r>0$ we have
\begin{equation}
R^*_F(\psi\to\EPR_{AB},r)=\sup_{0\le x\le r}\left[R^*(\psi\to\EPR_{AB},x)+r-x\right].
\end{equation}
\end{proposition}
\begin{proof}
We show that the right hand side is a lower bound on the left hand side. Let $0\le x\le r$, and let $(p_m)_{m\in\naturals}$, $(n_m)_{m\in\naturals}$ such that $\lim_{m\to\infty}-\frac{1}{m}\log p_m=x$ and $\lim_{m\to\infty}\frac{n_m}{m}=R^*(\psi\to\EPR_{AB},x)$ and $\ketbra{\psi}{\psi}^{\otimes m}\loccto p_m\ketbra{\EPR_{AB}}{\EPR_{AB}}^{\otimes n_m}$ for all $m$. With $R=R^*(\psi\to\EPR_{AB},x)+r-x$, $\epsilon_m=1-p_m\frac{2^{n_m}}{2^{\lfloor Rm\rfloor}}$, for all $m$ we have
\begin{equation}
\ketbra{\psi}{\psi}^{\otimes m}\loccto[\epsilon_m]\EPR_{AB}^{\otimes\lfloor Rm\rfloor}
\end{equation}
by \cref{lem:lowprobabilitytofidelity}. Since
\begin{equation}
\begin{split}
\lim_{m\to\infty}-\frac{1}{m}\log(1-\epsilon_m)
 & = \lim_{m\to\infty}-\frac{1}{m}\log(p_m\frac{2^{n_m}}{2^{\lfloor Rm\rfloor}})  \\
 & = x-R^*(\psi\to\EPR_{AB},x)+R  \\
 & = x-R^*(\psi\to\EPR_{AB},x)+R^*(\psi\to\EPR_{AB},x)+r-x
   = r
\end{split}
\end{equation}
and $\lfloor Rm\rfloor/m\to R$, we have
\begin{equation}
R^*_F(\psi\to\EPR_{AB},r)\ge R=R^*(\psi\to\EPR_{AB},x)+r-x.
\end{equation}

We show that the left hand side is a lower bound on the right hand side. Let $0<r$, $(\epsilon_m)_{m\in\naturals}$ and $(n_m)_{m\in\naturals}$ such that $\ketbra{\psi}{\psi}^{\otimes m}\loccto[\epsilon_m]\EPR^{\otimes n_m}$ holds for all $m$, $\lim_{m\to\infty}-\frac{1}{m}(1-\epsilon_m)=r$ and $\lim_{m\to\infty}\frac{n_m}{m}=R^*_F(\psi\to\EPR_{AB},r)$. By the first part, $R^*_F(\psi\to\EPR_{AB},r)\ge r$ holds. If this is an equality, then with $x=0$ the inequality $R^*_F(\psi\to\EPR_{AB},r)\le R^*(\psi\to\EPR_{AB},x)+r-x$ holds. We may therefore assume that $R^*_F(\psi\to\EPR_{AB},r)>r$.

By \cref{lem:lowfidelitytoprobability}, for all $m$ there exist $L_m\le 2^{n_m}$ and $P_m\in[0,1]$ such that $\ketbra{\psi}{\psi}^{\otimes m}\loccto P_m\ketbra{\phi_{L_m}}{\phi_{L_m}}\loccto P_m\EPR_{AB}^{\otimes\lfloor \log L_m\rfloor}$ and the inequality
\begin{equation}
\sqrt{P_mL_m}\ge\frac{\sqrt{2^{n_m}(1-\epsilon_m)}-1}{(\ln 2)n_m}
\end{equation}
holds. By \cref{prop:suprates}, with $x_m=-\frac{1}{m}\log P_m$ we have for all $m$
\begin{equation}
R^*(\psi\to\EPR_{AB},x_m)\ge\frac{1}{m}\lfloor\log L_m\rfloor.
\end{equation}
It follows that
\begin{equation}
\begin{split}
\liminf_{m\to\infty}R_F(\psi\to\EPR_{AB},x_m)-x_m
 & \ge \liminf_{m\to\infty}\frac{1}{m}\lfloor\log L_m\rfloor-x_m  \\
 & = \liminf_{m\to\infty}\frac{1}{m}\log P_mL_m  \\
 & = \liminf_{m\to\infty}\frac{1}{m}\log\left(\frac{\sqrt{2^{n_m}(1-\epsilon_m)}-1}{(\ln 2)n_m}\right)^2  \\
 & = R^*_F(\psi\to\EPR_{AB},r)-r.
\end{split}
\end{equation}
If the sequence $x_m$ has an accumulation point in $[0,r)$, then this implies that $R^*_F(\psi\to\EPR_{AB},r)\le\sup_{0\le x< r}\left[R^*(\psi\to\EPR_{AB},x)+r-x\right]$. Otherwise note that
\begin{equation}
\begin{split}
P_m=\frac{P_mL_m}{L_m}\ge\frac{P_mL_m}{2^{n_m}}\ge\left(\frac{\sqrt{2^{n_m}(1-\epsilon_m)}-1}{(\ln 2)n_m}\right)^2,
\end{split}
\end{equation}
therefore
\begin{equation}
\limsup_{m\to\infty}x_m\le r,
\end{equation}
which implies that this is a limit. $R^*(\psi\to\EPR_{AB},\cdot)$ is concave and therefore continuous in $\positivereals$, therefore in this case we obtain $R^*_F(\psi\to\EPR_{AB},r)\le\left[R^*(\psi\to\EPR_{AB},r)+r-r\right]$.
\end{proof}

\subsection{Types}

An $n$-type on a finite set $\mathcal{X}$ is a probability distribution that assigns probabilities that are multiples of $\frac{1}{n}$ to each element. The set of $n$-types will be denoted by $\distributions[n](\mathcal{X})$. For $P\in\distributions[n](\mathcal{X})$, the type class $\typeclass{n}{P}$ is the set of strings in $\mathcal{X}^n$ in which each symbol $x$ occurs $nP(x)$ times. If $\ket{x}_{x\in\mathcal{X}}$ is a preferred basis in a Hilbert space $\mathcal{H}$, then the type class projector is
\begin{equation}
\typeclassprojector{n}{P}\sum_{(x_1,\dots,x_n)\in\typeclass{n}{P}}\ketbra{x_1}{x_1}\otimes\dots\otimes\ketbra{x_n}{x_n}\in\boundeds(\mathcal{H}^{\otimes n}).
\end{equation}

For an $n$-type $P$ on a product set $\mathcal{X}\times\mathcal{Y}$ its marginals $P_X(x)=\sum_{y\in\mathcal{Y}}P(x,y)$ and $P_Y(y)=\sum_{x\in\mathcal{X}}P(x,y)$ are $n$-types as well, and the relation $\typeclass{n}{P}\subseteq\typeclass{n}{P_X}\times\typeclass{n}{P_Y}$ holds (under the canonical identification $(\mathcal{X}\times\mathcal{Y})^n\simeq\mathcal{X}^n\times\mathcal{Y}^n$), and accordingly $\typeclassprojector{n}{P}\le\typeclassprojector{n}{P_X}\otimes\typeclassprojector{n}{P_Y}$.

The number of $n$-types over $\mathcal{X}$ is polynomially bounded as $\lvert\distributions[n](\mathcal{X})\rvert\le(x+1)^{\lvert\mathcal{X}\rvert}$. The cardinalities of the type classes satisfy
\begin{equation}
\frac{1}{(n+1)^{\lvert\mathcal{X}\rvert}}2^{n\entropy(Q)}\le\lvert\typeclass{n}{Q}\rvert\le 2^{n\entropy(Q)}.
\end{equation}
More generally, if $P$ is a measure on $\mathcal{X}$, then
\begin{equation}
\frac{1}{(n+1)^{\lvert\mathcal{X}\rvert}}2^{-n\relativeentropy{Q}{P}}\le P^{\otimes n}(\typeclass{n}{Q})\le 2^{-n\relativeentropy{Q}{P}}.
\end{equation}

Let $\partitions[n]$ denote the set of partitions of the integer $n$, identified with nonincreasing nonnegative integer sequences, with traing zeros added or removed whenever convenient. By the Schur--Weyl decomposition, for any finite-dimensional Hilbert space $\mathcal{H}$ and $n\in\naturals$ the space $\mathcal{H}^{\otimes n}$ is isomorphic to
\begin{equation}
\mathcal{H}^{\otimes n}\simeq\bigoplus_{\lambda\in\partitions[n]}\mathbb{S}_\lambda(\mathcal{H})\otimes[\lambda]
\end{equation}
as representations of $U(\mathcal{H})\times S_n$, where $\mathbb{S}_\lambda(\mathcal{H})$ is an irreducible representation of the unitary group $U(\mathcal{H})$ if $\lambda$ has at most $\dim\mathcal{H}$ nonzero entries and the zero representation otherwise, while $[\lambda]$ is an irreducible representation of the symmetric group $S_n$. We have the dimension estimates
\cite[eqs. (6.16) and (6.21)]{hayashi2017group} (see also \cite{hayashi2002exponents,christandl2006spectra,harrow2005applications})
\begin{equation}\label{eq:unitarydimensionbound}
\dim\mathbb{S}_\lambda(\mathcal{H})\le(n+1)^{d(d-1)/2}
\end{equation}
\begin{equation}\label{eq:symmetricdimensionbound}
\frac{1}{(n+d)^{(d+2)(d-1)/2}}2^{n\entropy(\lambda/n)}\le \dim[\lambda]\le 2^{n\entropy(\lambda/n)}.
\end{equation}
We denote the projection onto the subspace corresponding to $\lambda$ by $P^{\mathcal{H}}_\lambda$ or by $P_\lambda$ when the Hilbert space is understood.

A key ingredient in the construction of the entanglement measures $E^{\alpha,\theta}$ in \cite{vrana2023family} is the large deviation rate function for the simultaneous empirical Young diagram measurement. For this we introduce simplified notation for tuples of partitions and normalized versions thereof. With $\overline{\partitions}$ we denote the set of nonincreasing nonnegative real-valued finite sequences that summing to $1$ (again, up to adding or removing trailing zeros). A useful way to view these is as (limits of) normalized partitions: if $\lambda\in\partitions[n]$, then $\frac{1}{n}\lambda\in\overline{\partitions}$.

A finite sequence $x=(x_1,x_2,\dots,x_d)$ majorizes $y=(y_1,y_2,\dots,y_d)$ (denoted $x\succcurlyeq y$), both assumed to be ordered nonincreasingly, if
\begin{equation}
\sum_{i=1}^m x_i\ge\sum_{i=1}^m y_i
\end{equation}
for $m=1,\dots,d$, with equality for $m=d$. In the context of integer partitions, this relation is usually known as the dominance order. By applying to the decreasingly ordered and padded versions, we extend majorization to arbitrary finite sequences of real numbers. A basic property is that if $P\in\distributions(\mathcal{X})$ and $Q\in\distributions(\mathcal{Y})$, then $P\preccurlyeq Q$ implies $\entropy(P)\ge\entropy(Q)$.

In what follows, $\lambda$ will denote a \emph{triple} of partitions $\lambda=(\lambda_A,\lambda_B,\lambda_C)$ of the same integer $n$, and likewise $\overline{\lambda}=(\overline{\lambda}_A,\overline{\lambda}_B,\overline{\lambda}_C)\in\overline{\partitions}^3$. We write $\ball{\epsilon}{\overline{\lambda}}$ for the $\epsilon$-ball around $\overline{\lambda}$. The precise distance measure is not particularly important, but for concreteness we take the maximum of the $\ell^1$ norms, i.e. $\overline{\mu}\in\ball{\epsilon}{\overline{\lambda}}$ iff $\max\{\norm[1]{\overline{\mu}_A-\overline{\lambda}_A},\norm[1]{\overline{\mu}_B-\overline{\lambda}_B},\norm[1]{\overline{\mu}_C-\overline{\lambda}_C}\}<\epsilon$.

Triples of partitions are used for indexing products of Schur--Weyl projections. For $\lambda=(\lambda_A,\lambda_B,\lambda_C)\in\partitions[n]^2$ we set $P_\lambda=P^{\mathcal{H}_A}_{\lambda_A}\otimes P^{\mathcal{H}_B}_{\lambda_B}\otimes P^{\mathcal{H}_C}_{\lambda_C}$, which acts on $(\mathcal{H}_A\otimes\mathcal{H}_B\otimes\mathcal{H}_C)^{\otimes n}$. With these notations, for a state vector $\psi\in\mathcal{H}_A\otimes\mathcal{H}_B\otimes\mathcal{H}_C$ and $\overline{\lambda}\in\overline{\partitions}^3$ we define
\begin{equation}
I_\psi(\overline{\lambda})=\lim_{\epsilon\to 0}\lim_{n\to\infty}-\frac{1}{n}\log\sum_{\substack{\lambda\in\partitions[n]^3  \\  \frac{1}{\lambda}\in\ball{\epsilon}{\overline{\lambda}}}}\norm{P_\lambda\psi^{\otimes n}}^2.
\end{equation}

The entanglement measures $E^{\alpha,\theta}$ for $\alpha\in(0,1)$ and $\theta\in\distributions(\{A,B,C\})$ are defined as
\begin{equation}
E^{\alpha,\theta}(\psi)=\sup_{\overline{\lambda}\in\overline{\partitions}^3}\left[\theta(A)\entropy(\overline{\lambda}_A)+\theta(B)\entropy(\overline{\lambda}_B)+\theta(C)\entropy(\overline{\lambda}_C)-\frac{\alpha}{1-\alpha}I_\psi(\overline{\lambda})\right].
\end{equation}

\section{Deterministic transformations}\label{sec:deterministic}

We start by recalling a tail bound for rectangular matrix Rademacher series.
\begin{proposition}[{\cite[Corollary 4.2]{tropp2012user}}]\label{prop:Rademachertailbound}
Let $Z_1,\dots,Z_d$ be $d_1\times d_2$ matrices and $\gamma_1,\gamma_2,\dots,\gamma_d$ independent Rademacher random variables (i.e., uniformly distributed on $\{-1,1\}$). With
\begin{equation}
\sigma^2:=\max\left\{\norm[\infty]{\sum_{k=1}^d Z_kZ_k^*},\norm[\infty]{\sum_{k=1}^d Z_k^*Z_k}\right\}
\end{equation}
the inequality
\begin{equation}
\probability\left(\norm[\infty]{\sum_{k=1}^d \gamma_kZ_k}\ge t\right)\le(d_1+d_2)e^{-\frac{1}{2}\frac{t^2}{\sigma^2}}
\end{equation}
holds for all $t\ge 0$.
\end{proposition}

\begin{proposition}\label{prop:goodvector}
Let $\ket{\psi}\in\mathcal{H}_A\otimes\mathcal{H}_B\otimes\mathcal{H}_C$ be a pure state vector, and let $\{\ket{k}\}_{k\in[d_C]}$ be a basis that diagonalizes $\Tr_{AB}\ketbra{\psi}{\psi}$. Let $\gamma_1,\gamma_2,\dots,\gamma_{d_C}$ be independent Rademacher random variables and $\ket{v}=\sum_{k=1}^{d_C}\gamma_k\ket{k}$, and consider the random vector $\ket{\phi}=(I\otimes I\otimes\bra{v})\ket{\psi}\in\mathcal{H}_A\otimes\mathcal{H}_B$. Then $\norm{\phi}=1$ almost surely, and
\begin{equation}
\probability\left(E_\infty(\phi)\le h\right)\le(d_A+d_B)e^{-\frac{1}{2}2^{\min\{\entropy_\infty(A)_\psi,\entropy_\infty(B)_\psi\}-h}}.
\end{equation}
\end{proposition}
\begin{proof}
The norm can be written as
\begin{equation}
\begin{split}
\norm{\phi}^2
 & = \braket{\phi}{\phi}  \\
 & = \bra{\psi}(I\otimes I\otimes\ketbra{v}{v})\ket{\psi}  \\
 & = \Tr\ketbra{v}{v}\Tr_{AB}\ketbra{\psi}{\psi}  \\
 & = \sum_{k=1}^d\lvert\gamma_k\rvert^2\bra{k}\left(\Tr_{AB}\ketbra{\psi}{\psi}\right)\ket{k},
\end{split}
\end{equation}
since $\Tr_{AB}\ketbra{\psi}{\psi}$ is diagonal in this basis by assumption. As the distribution of $\gamma_i$ is concentrated at $\pm 1$, the sum is equal to $\Tr\Tr_{AB}\ketbra{\psi}{\psi}=1$ with probability $1$.

Choose arbitrary orthonormal bases in $\mathcal{H}_A$ and $\mathcal{H}_B$, and consider the $d_A\times d_B$ matrices $(Z_k)_{ij}:=(\bra{i}\otimes\bra{j}\otimes\bra{k})\ket{\psi}$ and $Z_{ij}=(\bra{i}\otimes\bra{j})\ket{\phi}$, i.e., $Z=\sum_{i=k}^{d_C}\gamma_k Z_k$. The matrices $Z_1,\dots,Z_{d_C}$ may be pictured as the ``slices'' of $\ket{\psi}$ and $Z$ is a matrix representation of the pure bipartite state vector $\ket{\phi}$. Then
\begin{equation}
\begin{split}
\Tr_B\ketbra{\phi}{\phi}
 & = \sum_{i,i'=1}^{d_A}\sum_{j=1}^{d_B}(\ketbra{i}{i}\otimes\bra{j})\ketbra{\phi}{\phi}(\ketbra{i'}{i'}\otimes\ket{k})  \\
 & = \sum_{i,i'=1}^{d_A}\sum_{j=1}^{d_B}Z_{ij}\overline{Z_{ji'}}\ketbra{i}{i'}  \\
 & = \sum_{i,i'=1}^{d_A}(ZZ^*)_{ii'}\ketbra{i}{i'},
\end{split}
\end{equation}
therefore
\begin{equation}
\norm[\infty]{\Tr_B\ketbra{\phi}{\phi}}=\norm[\infty]{ZZ^*}=\norm[\infty]{Z}^2=\norm[\infty]{\sum_{k=1}^{d_C}\gamma_k Z_k}^2.
\end{equation}
In a similar way,
\begin{equation}
\begin{split}
\Tr_{BC}\ketbra{\psi}{\psi}
 & = \sum_{i,i'=1}^{d_A}\sum_{j=1}^{d_B}\sum_{k=1}^{d_C}(\ketbra{i}{i}\otimes\bra{j}\otimes\bra{k})\ketbra{\psi}{\psi}(\ketbra{i'}{i'}\otimes\ket{j}\otimes\ket{k})  \\
 & = \sum_{i,i'=1}^{d_A}\sum_{j=1}^{d_B}\sum_{k=1}^{d_C} (Z_k)_{ij}\overline{(Z_k)_{i'j}}\ketbra{i}{i'}  \\
 & = \sum_{k=1}^{d_C}\sum_{i,i'=1}^{d_A}(Z_kZ_k^*)_{ii'}\ketbra{i}{i'}
\end{split}
\end{equation}
and
\begin{equation}
\begin{split}
\Tr_{AC}\ketbra{\psi}{\psi}
 & = \sum_{i=1}^{d_A}\sum_{j,j'=1}^{d_B}\sum_{k=1}^{d_C}(\bra{i}\otimes\ketbra{j}{j}\otimes\bra{k})\ketbra{\psi}{\psi}(\ket{i}\otimes\ketbra{j'}{j'}\otimes\ket{k})  \\
 & = \sum_{i=1}^{d_A}\sum_{j,j'=1}^{d_B}\sum_{k=1}^{d_C} (Z_k)_{ij}\overline{(Z_k)_{ij'}}\ketbra{i}{j'}  \\
 & = \sum_{k=1}^{d_C}\sum_{j,j'=1}^{d_A}\overline{(Z_k^*Z_k)_{jj'}}\ketbra{j}{j'},
\end{split}
\end{equation}
therefore
\begin{align}
\norm[\infty]{\Tr_{BC}\ketbra{\psi}{\psi}} & = \norm[\infty]{\sum_{k=1}^{d_C}Z_kZ_k^*}  \\
\norm[\infty]{\Tr_{AC}\ketbra{\psi}{\psi}} & = \norm[\infty]{\sum_{k=1}^{d_C}Z_k^*Z_k}.
\end{align}

In terms of entropies,
\begin{equation}
\begin{split}
\sigma^2
 & := \max\left\{\norm[\infty]{\sum_{k=1}^d Z_kZ_k^*},\norm[\infty]{\sum_{k=1}^d Z_k^*Z_k}\right\}  \\
 & = \max\left\{2^{-\entropy_\infty(A)_\psi},2^{-\entropy_\infty(B)_\psi}\right\}  \\
 & = 2^{-\min\{\entropy_\infty(A)_\psi,\entropy_\infty(B)_\psi\}},
\end{split}
\end{equation}
and $E_\infty(\phi)\le h$ is equivalent to
\begin{equation}
2^{-h}\le\norm[\infty]{\Tr_B\ketbra{\phi}{\phi}}=\norm[\infty]{\sum_{k=1}^{d_C}\gamma_k Z_k}^2
.\end{equation}
From \cref{prop:Rademachertailbound} we conclude
\begin{equation}
\begin{split}
\probability\left(E_\infty(\phi)\le h\right)
 & = \probability\left(\norm[\infty]{\sum_{k=1}^{d_C}\gamma_k Z_k}\ge 2^{-h/2}\right)  \\
 & \le (d_A+d_B)e^{-\frac{1}{2}\frac{2^{-h}}{\sigma^2}}  \\
 & = (d_A+d_B)e^{-\frac{1}{2}2^{\min\{\entropy_\infty(A)_\psi,\entropy_\infty(B)_\psi\}-h}}.
\end{split}
\end{equation}
\end{proof}

In particular, \cref{prop:goodvector} implies that there exists a vector $\ket{v}$ such that the local projection of $\ket{\psi}$ onto the subspace spanned by $\ket{v}$ is highly entangled: any choice with $h<\min\{\entropy_\infty(A)_\psi,\entropy_\infty(B)_\psi\}-\log\left(2\ln(d_A+d_B)\right)$ ensures
\begin{equation}
(d_A+d_B)e^{-\frac{1}{2}2^{\min\{\entropy_\infty(A)_\psi,\entropy_\infty(B)_\psi\}-h}}<1.
\end{equation}
To construct a rank-$1$ POVM with the property that every post-measurement state has a high entropy, we need to find many such vectors, however. More precisely, given a family of vectors $(\ket{v_i})_{i\in I}$ in $\mathcal{H}_C$, suitable multiples of the rank-$1$ operators $\ketbra{v_i}{v_i}$ form a POVM if and only if the cone that they generate contains $I$.
\begin{remark}
If $\ket{\psi}$ is sufficiently symmetric, then finding a single vector $\ket{v}$ is sufficient. More precisely, if a group $G$ acts unitarily on $\mathcal{H}_A$, $\mathcal{H}_B$ and $\mathcal{H}_C$ and $\ketbra{\psi}{\psi}$ is $G$-invariant, then the orbit of $\ket{v}$ gives rise to states on $AB$ with equal min-entropy of entanglement. If in addition the group is compact and the representation on $\mathcal{H}_C$ is irreducible, then the Haar average of the orbit of $\ketbra{v}{v}$ is a multiple of $I$.
\end{remark}

In the absence of symmetries, a possible strategy is to sample several vectors independently from the distribution described in \cref{prop:goodvector}. Since the uniform convex combination of \emph{all} rank-$1$ matrices with $\pm 1$ entries is $I$, by sampling sufficiently many vectors $\ket{v_1},\ket{v_2},\dots,\ket{v_n}$, we can ensure that with high probability the identity is in the convex hull of the operators $\ketbra{v_1}{v_1},\dots,\ketbra{v_n}{v_n}$. On the other hand, sampling too many vectors increases the probability that at least one of them leads to a measurement outcome with low entanglement, which is not suitable for deterministic transformations. For this reason, we need a good estimate on the number of samples needed to form a POVM. As a crude lower bound, we clearly need at least $d_C$ vectors, otherwise there are no full-rank operators in the cone.

To find an upper bound, we will use a recent result of Hayakawa, Lyons and Oberhauser \cite{hayakawa2023estimating}. In a general setting, $n$ vectors $X_1,\dots,X_n$ are drawn from a probability distribution $\mu$ on $\reals^m$. Given a point $\theta\in\reals^m$, one defines $N_\mu(\theta)$ as the smallest $n$ such that $\probability(\theta\in\convexhull\{X_1,\dots,X_n\})\ge\frac{1}{2}$. To state the bound on $N_\mu(\theta)$, we need the following definition, proposed by Tukey \cite{tukey1974mathematics} for finite sets of points and later generalized to probability distributions and extensively studied in \cite{rousseeuw1999depth}.
\begin{definition}
Let $\mu$ be a probability measure on $\reals^m$ and $\theta\in\reals^m$. The halfspace depth of $\theta$ is
\begin{equation}
\alpha_\mu(\theta)=\inf_{c\in\reals^m\setminus\{0\}}\mu(\setbuild{x\in\reals^m}{\langle c,x-\theta\rangle\le 0}).
\end{equation}
\end{definition}
\begin{theorem}[{\cite[Theorem 16]{hayakawa2023estimating}}]\label{thm:depthbound}
For every probability distribution $\mu$ on $\reals^m$ and point $\theta\in\reals^m$ we have
\begin{equation}
N_\mu(\theta)\le\left\lceil\frac{3m}{\alpha_\mu(\theta)}\right\rceil.
\end{equation}
\end{theorem}
While computing $\alpha_\mu(\theta)$ does not seem to be easy in general, fortunately in our case basic general properties provide sufficiently good bounds, thanks to the symmetry of the distribution.
\begin{proposition}[{\cite[Proposition 1]{rousseeuw1999depth}}]
For every $\mu$ the function $\theta\mapsto\alpha_\mu(\theta)$ is quasi-concave.
\end{proposition}
\begin{proposition}[{\cite[Proposition 9]{rousseeuw1999depth}}]
For every probability distribution $\mu$ on $\reals^m$ there exists $\theta\in\reals^m$ such that $\alpha_\mu(\theta)\ge\frac{1}{m+1}$.
\end{proposition}

We can now prove our one-shot deterministic bound for tripartite to bipartite entanglement transformations.
\begin{theorem}\label{thm:oneshotdeterministic}
Let $\ket{\psi}\in\mathcal{H}_A\otimes\mathcal{H}_B\otimes\mathcal{H}_C$ be a pure state vector. Then $\ketbra{\psi}{\psi}\loccto\EPR_{AB}^{\otimes n}$ where
\begin{equation}
n=\left\lceil \min\{\entropy_\infty(A)_\psi,\entropy_\infty(B)_\psi\}-\log\left(4\ln(2d_C^4(d_A+d_B))\right)
\right\rceil.
\end{equation}
\end{theorem}
\begin{proof}
Let $\mu$ be the distribution of the random vectors (operators) $X_i=\ketbra{v_i}{v_i}$, where the $\ket{v_i}$ are uniformly drawn from $\{1,-1\}^{d_C}\in\complexes^{d_C}\simeq\mathcal{H}_C$, where the identification is through any basis that diagonalizes $\Tr_{AB}\ketbra{\psi}{\psi}$, as in \cref{prop:goodvector}. Since the diagonal of $X_i$ almost surely consists of only $1$ entries, and all the other entries are real, the real dimension of the affine hull of the support of $\mu$ is $m=d_C(d_C-1)/2$. By the above properties, there is a matrix $\theta$ with all-$1$ diagonal such that $\alpha_\mu(\theta)\ge\frac{1}{d_C(d_C-1)/2+1}$. The distribution of $X_i$ is invariant under conjugations by diagonal orthogonal matrices, therefore the same lower bound applies to $\alpha_\mu(O\theta)$ for every diagonal orthogonal matrix $O$. By the quasi-concavity, we have
\begin{equation}
\begin{split}
\frac{1}{d_C(d_C-1)/2+1}
 & \le \min_{O}\alpha_\mu(O\theta)  \\
 & \le \alpha_\mu\left(\frac{1}{2^{d_C}}\sum_{O}O\theta\right)  \\
 & = \alpha_\mu(I).
\end{split}
\end{equation}
From \cref{thm:depthbound} we conclude
\begin{equation}
N_\mu(I)
 \le \left\lceil\frac{3d_C(d_C-1)/2}{\alpha_\mu(I)}\right\rceil
 \le \left\lceil3\frac{d_C(d_C-1)}{2}\left(\frac{d_C(d_C-1)}{2}+1\right)\right\rceil
 \le d_C^4.
\end{equation}

Draw $d_C^4$ vectors $\ket{v_1},\dots,\ket{\smash{v_{d_C^4}}}$ from the distribution described in \cref{prop:goodvector} and let $\ket{\phi_i}=(I\otimes I\otimes\bra{v_i})\ket{\psi}$. By the union bound,
\begin{equation}
\probability(\exists i\in[d_C^4]:E_\infty(\ket{\phi_i})\le h)\le d_C^4(d_A+d_B)e^{-\frac{1}{2}2^{\min\{\entropy_\infty(A)_\psi,\entropy_\infty(B)_\psi\}-h}},
\end{equation}
which is less than $\frac{1}{2}$ if
\begin{equation}
h<\min\{\entropy_\infty(A)_\psi,\entropy_\infty(B)_\psi\}-\log\left(2\ln(2d_C^4(d_A+d_B))\right).
\end{equation}
In particular, this holds for $h=n$.

By the definition of $N_\mu(I)$, with probability at least $1/2$ the identity matrix is in the convex hull of $X_1,\dots,X_{d_C^4}$. By the union bound, there is at least one realization of the random variables such that both conditions are satisfied. If $\sum_{i=1}^{d_C^4}\lambda_iX_i=I$ is such a convex combination, then the operators $E_i:=\lambda_iX_i$ form a POVM on $\mathcal{H}_C$ with the property that the min-entropy of entanglement of every post-measurement state $(I\otimes I\otimes E_i)\ket{\psi}$ is at least $n$. The protocol now consists of $C$ performing the measurement described by the POVM, communicating the result to $A$ and $B$, who can transform $\ket{\psi_i}$ into $n$ EPR pairs by Nielsen's theorem.
\end{proof}

\begin{corollary}\label{cor:deterministicrate}
$\displaystyle R_{\textnormal{det}}(\psi\to\EPR_{AB})=\min\{\entropy_\infty(A)_\psi,\entropy_\infty(B)_\psi\}$.
\end{corollary}
\begin{proof}
The right hand side is an upper bound on the left hand side since $\entropy_\infty(A)$ is the optimal deterministic rate when joint operations on $BC$ are also allowed, and likewise, $\entropy_\infty(B)$ is the optimal deterministic rate for LOCC operations with respect to the bipartition $B:AC$.

For the achievability we apply \cref{thm:oneshotdeterministic} to $\ket{\psi}^{\otimes m}$:
\begin{equation}
\begin{split}
R_{\textnormal{det}}(\psi\to\EPR_{AB})
 & \ge \lim_{m\to\infty}\frac{1}{m}\left(\left\lceil \min\{\entropy_\infty(A)_{\psi^{\otimes m}},\entropy_\infty(B)_{\psi^{\otimes m}}\}-\log\left(4\ln(2d_C^{4m}(d_A^m+d_B^m))\right)\right\rceil\right)  \\
 & = \lim_{m\to\infty}\frac{1}{m}\left(\left\lceil m\min\{\entropy_\infty(A)_\psi,\entropy_\infty(B)_\psi\}-\log\left(4\ln(2d_C^{4m}(d_A^m+d_B^m))\right)\right\rceil\right)  \\
 & \ge \lim_{m\to\infty}\frac{1}{m}\left(\left\lceil m\min\{\entropy_\infty(A)_\psi,\entropy_\infty(B)_\psi\}-\log\left(4m\ln(4d_C^{4}\max\{d_A,d_B\})\right)\right\rceil\right)  \\
 & = \min\{\entropy_\infty(A)_\psi,\entropy_\infty(B)_\psi\}.
\end{split}
\end{equation}
\end{proof}

\section{Direct exponents}\label{sec:direct}

In order to prove the achievability for the direct exponents, we extend the truncation method of \cite{hayashi2002error} by combining two truncating measurements acting on different subsystems. We use the following lemma.
\begin{lemma}\label{lem:localcontraction}
Let $M\in\boundeds(\mathcal{H}_A)$, $\norm[\infty]{M}\le 1$, and $\rho\in\states(\mathcal{H}_A\otimes\mathcal{H}_B)$. Then $\Tr_A (M\otimes I)\rho(M^*\otimes I)\le\Tr_A\rho$.
\end{lemma}
\begin{proof}
Let $X\in\boundeds(\mathcal{H}_B)_{\ge 0}$. Then $M^*M\otimes X\le I\otimes X$, therefore
\begin{equation}
\begin{split}
\Tr X\Tr_A (M\otimes I)\rho(M^*\otimes I)
 & = \Tr (I\otimes X)(M\otimes I)\rho(M^*\otimes I)  \\
 & = \Tr (M^*M\otimes X)\rho  \\
 & \le \Tr (I\otimes X)\rho  \\
 & = \Tr X\Tr_A \rho.
\end{split}
\end{equation}
\end{proof}

\begin{proposition}\label{prop:simultaneoustruncation}
Let $\ket{\psi}\in\mathcal{H}_A\otimes\mathcal{H}_B\otimes\mathcal{H}_C$ be a pure state vector and let $M_A\in\boundeds(\mathcal{H}_A)$ and $M_B\in\boundeds(\mathcal{H}_B)$ be contractions. The success probabilities
\begin{align}
p_{\textnormal{s},A} & := \norm{(M_A\otimes I\otimes I)\psi}^2  \\
p_{\textnormal{s},B} & := \norm{(I\otimes M_B\otimes I)\psi}^2  \\
p_{\textnormal{s},AB} & := \norm{(M_A\otimes M_B\otimes I)\psi}^2
\end{align}
and post-measurement states
\begin{align}
\psi'_{A} & := \frac{(M_A\otimes I\otimes I)\psi}{\sqrt{p_{\textnormal{s},A}}}  \\
\psi'_{B} & := \frac{(I\otimes M_B\otimes I)\psi}{\sqrt{p_{\textnormal{s},B}}}  \\
\psi'_{AB} & := \frac{(M_A\otimes M_B\otimes I)\psi}{\sqrt{p_{\textnormal{s},AB}}}
\end{align}
satisfy the inequalities
\begin{align}
p_{\textnormal{s},AB} & \ge 1-(1-p_{\textnormal{s},A})-(1-p_{\textnormal{s},B})  \\
\entropy_\infty(A)_{\psi'_{AB}} & \ge \entropy_\infty(A)_{\psi'_{A}}-\log\frac{p_{\textnormal{s},A}}{p_{\textnormal{s},AB}}  \label{eq:truncationABA} \\
\entropy_\infty(B)_{\psi'_{AB}} & \ge \entropy_\infty(B)_{\psi'_{B}}-\log\frac{p_{\textnormal{s},B}}{p_{\textnormal{s},AB}}.  \label{eq:truncationABB}
\end{align}
\end{proposition}
\begin{proof}
The lower bound on $p_{\textnormal{s},AB}$ follows from the union bound.

By \cref{lem:localcontraction}, $\Tr_{BC}(M_A\otimes M_B\otimes I)\ketbra{\psi}{\psi}(M_A\otimes M_B\otimes I)^*\le \Tr_{BC}(M_A\otimes I\otimes I)\ketbra{\psi}{\psi}(M_A\otimes I\otimes I)^*$, therefore
\begin{equation}
\Tr_{BC}\ketbra{\psi'_{AB}}{\psi'_{AB}}\le \frac{p_{\textnormal{s},A}}{p_{\textnormal{s},AB}}\Tr_{BC}\ketbra{\psi'_{A}}{\psi'_{A}},
\end{equation}
which implies \eqref{eq:truncationABA}. The proof of \eqref{eq:truncationABB} is similar.
\end{proof}

\begin{theorem}\label{thm:directexponent}
Let $\ket{\psi}\in\mathcal{H}_A\otimes\mathcal{H}_B\otimes\mathcal{H}_C$ be a pure state vector. Then
\begin{equation}
R(\psi\to\EPR_{AB},r)=\min\left\{E_A(r),E_B(r)\right\},
\end{equation}
where
\begin{align}
E_A(r) & = \sup_{\alpha>1}\left[r\frac{1}{1-\alpha}+\entropy_\alpha(A)_\psi\right]  \\
E_B(r) & = \sup_{\alpha>1}\left[r\frac{1}{1-\alpha}+\entropy_\alpha(B)_\psi\right].
\end{align}
In particular, $R(\psi\to\EPR_{AB},r)=R_{\textnormal{det}}(\psi\to\EPR_{AB})$ if and only if $r\ge\min\{\entropy_\infty(A)_\psi,\entropy_\infty(B)_\psi\}$.
\end{theorem}
\begin{proof}
If we regard the state $\ket{\psi}_{ABC}$ as a bipartite state between $A$ and $BC$, then the optimal entanglement concentration rate with success probability at least $1-2^{-rn+o(n)}$ is $E_A(r)$ \cite[Corollary 7 and eq. (114)]{hayashi2002error}. Likewise, as a bipartite state between $B$ and $AC$, the optimal rate with the same probability is $E_B(r)$. It follows that $R(\psi\to\EPR_{AB},r)\le\min\left\{E_A(r),E_B(r)\right\}$.

To prove the achievability, we recall the method of \cite{hayashi2002error} for the bipartie case. There entanglement concentration is performed in two steps. The first, probabilistic step is a local two-outcome measurement designed to truncate the largest Schmidt coefficients, thereby increasing the min-entropy of entanglement with high probability. Introducing the functions
\begin{equation}
f_t(x)=\begin{cases}
\sqrt{\frac{t}{x}} & \text{if $t<x$}  \\
1 & \text{otherwise,}
\end{cases}
\end{equation}
the measurement operator on $n$ copies is $M_A^{(n)}=f_{t_A^{(n)}}(\Tr_{B^nC^n}\ketbra{\psi}{\psi}^{\otimes n})$ for a chosen threshold $t_A^{(n)}$. The second step is deterministic, and transforms the state to ebits. From Nielsen's theorem, the maximal number of ebits that can be extracted is the integer part of the min-entropy of entanglement.

An appropriate choice of the sequence $(t_A^{(n)})_{n=1}^\infty$ ensures that the success probability in the first step is $1-2^{-rn+o(n)}$ and the min-entropy of entanglement after the measurement is $nE_A(r)-o(n)$.

Likewise, for a suitable threshold sequence $(t_B^{(n)})_{n=1}^\infty$, the measurement operator $M_B^{(n)}=f_{t_B^{(n)}}(\Tr_{A^nC^n}\ketbra{\psi}{\psi}^{\otimes n})$ corresponds to an outcome with probability $1-2^{-rn+o(n)}$ and min-entropy of entanglement $nE_B(r)-o(n)$.

In the tripartite case, we apply the two probabilistic parts on subsystems $A$ and $B$ simultaneously. By \cref{prop:simultaneoustruncation}, the probability that both truncations are successful is also $1-2^{-rn+o(n)}$, and the local min-entropies at $A$ and $B$ are still $nE_A(r)-o(n)$ and $nE_B(r)-o(n)$, respectively. By \cref{thm:oneshotdeterministic}, the truncated state can be deterministically transformed into $\min\{nE_A(r)-o(n),nE_B(r)-o(n)\}$ EPR pairs between $A$ and $B$.
\end{proof}

We note that in the limit $r\to 0$ we get
\begin{equation}
\begin{split}
\lim_{r\to 0}R(\psi\to\EPR_{AB},r)
 & = \lim_{r\to 0}\min\left\{E_A(r),E_B(r)\right\}  \\
 & = \min\left\{\lim_{r\to 0}E_A(r),\lim_{r\to 0}E_B(r)\right\}  \\
 & = \min\left\{\entropy(A)_\psi,\entropy(B)_\psi\right\}.
\end{split}
\end{equation}
It follows that it is possible to transform copies of $\psi$ into EPR pairs between $A$ and $B$ with success probability $1-o(1)$ and rate $\min\left\{\entropy(A)_\psi,\entropy(B)_\psi\right\}$. This is optimal again because of the corresponding bipartite bounds, and provides a new proof of \cite[Theorem 1]{smolin2005entanglement}.

\section{Strong converse exponents}\label{sec:strongconverse}

We turn to the strong converse region, i.e., when the rate exceeds $\min\left\{\entropy(A)_\psi,\entropy(B)_\psi\right\}$ and the success probability approaches $0$ exponentially fast. In the case of bipartite entanglement concentration, the same truncation technique gives the optimal strong converse exponent as well \cite[Theorem 10]{hayashi2002error}. However, the method of \cref{thm:directexponent} does not work, because the lower bound on the probability of succesful simultaneous truncation in \cref{prop:simultaneoustruncation} becomes trivial when the success probability of both of the truncations is close to $0$. In fact, as we show in this section, the largest rate for a given strong converse exponent is in general not equal to the minimum of the bipartite rates $E^*_A(r)$ and $E^*_B(r)$, but is instead given by
\begin{equation}\label{eq:strongconverserateformula}
R^*(\psi\to\EPR_{AB},r)=\inf_{\substack{\alpha\in[0,1)  \\  x\in[0,1]}}\left[r\frac{\alpha}{1-\alpha}+E^{\alpha,(x,1-x,0)}(\psi)\right].
\end{equation}

\subsection{States with free support}

We start by analysing a special class of states, previously considered in \cite{bugar2024explicit}. The main ideas are similar to the general case, but this special case is technically simpler and in addition the resulting rate formula is more explicit. We recall the relevant definitions, specialised to tripartite states.
\begin{definition}
Let $\mathcal{X}_A$, $\mathcal{X}_B$, $\mathcal{X}_C$ be finite sets and $\Psi\subseteq\mathcal{X}_A\times\mathcal{X}_B\times\mathcal{X}_C$. We say that $\Psi$ is free if any two triples $x,y\in\Psi$ where $x\neq y$ are different in at least two positions, i.e., at most one of $x_A=y_A$, $x_B=y_B$, $x_C=y_C$ holds.

A state $\psi\in\mathcal{H}_A\otimes\mathcal{H}_B\otimes\mathcal{H}_C$ is said to have free support if there are local orthonormal bases $(\ket{x}_A)_{x\in\mathcal{X}_A}$, $(\ket{x}_B)_{x\in\mathcal{X}_B}$, $(\ket{x}_C)_{x\in\mathcal{X}_C}$ and a free subset $\Psi\subseteq\mathcal{X}_A\times\mathcal{X}_B\times\mathcal{X}_C$ such that
\begin{equation}
\psi=\sum_{(x_A,x_B,x_C)\in\Psi}\psi_{x_A,x_B,x_C}\ket{x_A}\otimes\ket{x_B}\otimes\ket{x_C}.
\end{equation}
\end{definition}
For instance, the weighted $W$ states $\sqrt{a}\ket{100}+\sqrt{b}\ket{010}+\sqrt{c}\ket{001}$ ($a+b+c=1$) have free support with respect to the standard three-qubit basis.

In this section, $\psi$ is always assumed to be a state with free support, and we will denote by $\mathcal{M}$ the measurement map with respect to a triple of local orthonormal bases as in the definition. Note that this measurement may not be unique, but the statements below hold for all of them (see \cite{bugar2024explicit} for more details). We regard the output of $\mathcal{M}$ as a classical probability distribution on $\mathcal{X}_A\times\mathcal{X}_B\times\mathcal{X}_C$. For probability distributions on product sets we consider the following quantities.
\begin{definition}
Let $\mathcal{X}_A$, $\mathcal{X}_B$, $\mathcal{X}_C$ be finite sets and let $P\in\distributions(\mathcal{X}_A\times\mathcal{X}_B\times\mathcal{X}_C)$, $\alpha\in[0,1)$, and $\theta\in\distributions(\{A,B,C\})$. We define
\begin{equation}
\entropy_{\alpha,\theta}(P)=\max_{Q\in\distributions(\support P)}\left[\theta(A)\entropy(Q_A)+\theta(B)\entropy(Q_B)+\theta(C)\entropy(Q_C)-\frac{\alpha}{1-\alpha}\relativeentropy{Q}{P}\right],
\end{equation}
where $Q_A,Q_B$ and $Q_C$ are the marginals of $Q$.
\end{definition}
The expression in the maximization is concave in $Q$, therefore $\entropy_{\alpha,\theta}(P)$ can be computed by convex optimization.
\begin{theorem}[{\cite[Theorem 4.7]{bugar2024explicit}}]
Let $\psi$ be a state with free support and $\mathcal{M}$ the measurement map as above. Then for all $\alpha\in[0,1)$ and $\theta\in\distributions(\{A,B,C\})$ we have $E^{\alpha,\theta}(\psi)=\entropy_{\alpha,\theta}(\mathcal{M}(\ketbra{\psi}{\psi}))$.
\end{theorem}

After these preparations, we first rewrite the upper bound on the rate in a different form.
\begin{proposition}\label{prop:freesupportupperbound}
Let $\psi$ be a state with free support and $\mathcal{M}$ the measurement map as above. Then
\begin{equation}
R^*(\psi\to\EPR_{AB},r)\le\max_{\substack{Q\in\distributions(\mathcal{X}_A\times\mathcal{X}_B\times\mathcal{X}_C)  \\  \relativeentropy{Q}{\mathcal{M}(\ketbra{\psi}{\psi})}\le r}}\min\{\entropy(Q_A),\entropy(Q_B)\}.
\end{equation}
\end{proposition}
\begin{proof}
For all $\alpha\in[0,1)$ and $\theta\in\distributions(\{A,B,C\})$ the inequality
\begin{equation}
\begin{split}
R^*(\psi\to\EPR_{AB},r)
 & \le \frac{r\frac{\alpha}{1-\alpha}+E^{\alpha,\theta}(\psi)}{E^{\alpha,\theta}(\EPR_{AB})}  \\
 & = \frac{r\frac{\alpha}{1-\alpha}+\entropy_{\alpha,\theta}(\mathcal{M}(\ketbra{\psi}{\psi}))}{\theta(A)+\theta(B)}
\end{split}
\end{equation}
holds. We set $\theta=(x,1-x,0)$ (i.e., $\theta(A)=x$, $\theta(B)=1-x$, $\theta(C)=0$) with $x\in[0,1]$, and minimize the right hand side over $\alpha$ and $x$:
\begin{equation}
\begin{split}
R^*(\psi\to\EPR_{AB},r)
 & \le \inf_{\substack{\alpha\in[0,1)  \\  x\in[0,1]}}\left[r\frac{\alpha}{1-\alpha}+\entropy_{\alpha,(x,1-x,0)}(\mathcal{M}(\ketbra{\psi}{\psi}))\right]  \\
 & = \inf_{\substack{\alpha\in[0,1)  \\  x\in[0,1]}}\max_Q\left[r\frac{\alpha}{1-\alpha}+x\entropy(Q_A)+(1-x)\entropy(Q_B)-\frac{\alpha}{1-\alpha}\relativeentropy{Q}{\mathcal{M}(\ketbra{\psi}{\psi})}\right]  \\
 & = \inf_{\substack{t\in[0,\infty)  \\  x\in[0,1]}}\max_Q\left[rt+x\entropy(Q_A)+(1-x)\entropy(Q_B)-t\relativeentropy{Q}{\mathcal{M}(\ketbra{\psi}{\psi})}\right]  \\
 & = \inf_{t\in[0,\infty)}\max_Q\left[rt+\min\{\entropy(Q_A),\entropy(Q_B)\}-t\relativeentropy{Q}{\mathcal{M}(\ketbra{\psi}{\psi})}\right]  \\
 & = \max_Q\inf_{t\in[0,\infty)}\left[rt+\min\{\entropy(Q_A),\entropy(Q_B)\}-t\relativeentropy{Q}{\mathcal{M}(\ketbra{\psi}{\psi})}\right]  \\
 & = \max_{\substack{Q  \\  \relativeentropy{Q}{\mathcal{M}(\ketbra{\psi}{\psi})}\le r}}\min\{\entropy(Q_A),\entropy(Q_B)\}.
\end{split}
\end{equation}
The first equality is by the definition of $\entropy_{\alpha,(x,1-x,0)}$, in the second equality the new variable $t=\frac{\alpha}{1-\alpha}$ is introduced, the third equality follows from Sion's minimax theorem \cite{sion1958general} (affine in $x$ and concave in $Q$) and that the minimum of an affine function on $[0,1]$ is attained at one of the endpoints, the fourth equality is another application of Sion's minimax theorem (affine in $x$ and concave in $Q$), and the last equality uses that the infimum over $t$ is either attained at $t=0$ (when $\relativeentropy{Q}{\mathcal{M}(\ketbra{\psi}{\psi})}\le r$) or is $-\infty$ (when $\relativeentropy{Q}{\mathcal{M}(\ketbra{\psi}{\psi})}>r$).
\end{proof}

It remains to show that the upper bound on the rate is achievable. Since the marginal entropies and the relative entropy are continuous on the set of probability distributions on the support of $\mathcal{M}(\ketbra{\psi}{\psi})$, it is sufficient to consider distributions with rational entries.
\begin{proposition}\label{prop:freesupportlowerbound}
Let $n\in\naturals$, $Q\in\distributions[n](\support\mathcal{M}(\ketbra{\psi}{\psi}))$, and $\relativeentropy{Q}{\mathcal{M}(\ketbra{\psi}{\psi})}\le r$. Then $R^*(\psi\to\EPR_{AB},r)\ge\min\{\entropy(Q_A),\entropy(Q_B)\}$.
\end{proposition}
\begin{proof}
For $m\in\naturals$, let $\Pi^{mn}_{Q_A}$, $\Pi^{mn}_{Q_B}$, $\Pi^{mn}_{Q_C}$ be the type class projections for the marginals and $\Pi^{mn}_Q$ the joint type class projector acting on $(\mathcal{H}_A\otimes\mathcal{H}_B\otimes\mathcal{H}_C)^{\otimes mn}$. $\psi^{\otimes mn}$ is has free support with respect to local orthonormal basis that diagonalize the marginal type class projectors, therefore the marginals of the state
\begin{equation}
\varphi:=\frac{(\Pi^{mn}_{Q_A}\otimes\Pi^{mn}_{Q_B}\otimes\Pi^{mn}_{Q_C})\psi^{\otimes mn}}{\norm{(\Pi^{mn}_{Q_A}\otimes\Pi^{mn}_{Q_B}\otimes\Pi^{mn}_{Q_C})\psi^{\otimes mn}}}
\end{equation}
are diagonal as well. They are also uniform since $S_{mn}$ acts transitively on the local basis vectors and $\varphi$ is symmetric. It follows from \cref{thm:oneshotdeterministic} that $\varphi$ can be transformed deterministically into at least
\begin{equation}
\min\{\log\lvert\typeclass{mn}{Q_A}\rvert,\log\lvert\typeclass{mn}{Q_B}\rvert\}-\log(4\ln(2d_C^{4mn}(d_A^{mn}+d_B^{mn})))
\end{equation}
EPR pairs between $A$ and $B$. The success probability of the overall transformation can be estimated as
\begin{equation}
\begin{split}
\norm{(\Pi^{mn}_{Q_A}\otimes\Pi^{mn}_{Q_B}\otimes\Pi^{mn}_{Q_C})\psi^{\otimes mn}}
 & \ge \norm{\Pi^{mn}_Q\psi^{\otimes mn}}  \\
 & \ge \frac{1}{(mn+1)^{d_Ad_Bd_C}}2^{-mn\relativeentropy{Q}{\mathcal{M}(\ketbra{\psi}{\psi})}}.
\end{split}
\end{equation}
We let $m\to\infty$ and use $\log\lvert\typeclass{mn}{Q_A}\rvert\ge(mn+1)^{-d_A}2^{mn\entropy(Q_A)}$ and $\log\lvert\typeclass{mn}{Q_B}\rvert\ge(mn+1)^{-d_B}2^{mn\entropy(Q_B)}$ to arrive at the stated bound.
\end{proof}

To summarize, \cref{prop:freesupportupperbound,prop:freesupportlowerbound} together prove the following.
\begin{theorem}\label{thm:freesupportscrate}
Let $\psi$ be a state with free support and $\mathcal{M}$ the measurement map as above. Then
\begin{equation}
R^*(\psi\to\EPR_{AB},r)=\max_{\substack{Q\in\distributions(\support\mathcal{M}(\ketbra{\psi}{\psi}))  \\  \relativeentropy{Q}{\mathcal{M}(\ketbra{\psi}{\psi})}\le r}}\min\{\entropy(Q_A),\entropy(Q_B)\}.
\end{equation}
\end{theorem}

As an example, the $W$ state $\ket{W}=\frac{1}{\sqrt{3}}(\ket{100}+\ket{010}+\ket{001})$ satisfies the support condition with respect to the standard basis. In \cite{bugar2024explicit} we numerically evaluated the upper bound
\begin{equation}
\inf_{\substack{\alpha\in[0,1)  \\  x\in[0,1]}}\left[r\frac{\alpha}{1-\alpha}+E^{\alpha,(x,1-x,0)}(W)\right]
 = \inf_{\alpha\in[0,1)}\left[r\frac{\alpha}{1-\alpha}+E^{\alpha,(1/2,1/2,0)}(W)\right],
\end{equation}
where the equality is due to the symmetry of the $W$ state and convexity of $E^{\alpha,\theta}$ in the weight $\theta$. \Cref{thm:freesupportscrate} implies that this is equal to the optimal rate $R^*(\W\to\EPR_{AB},r)$, shown in \cref{fig:WtoEPRrates}.
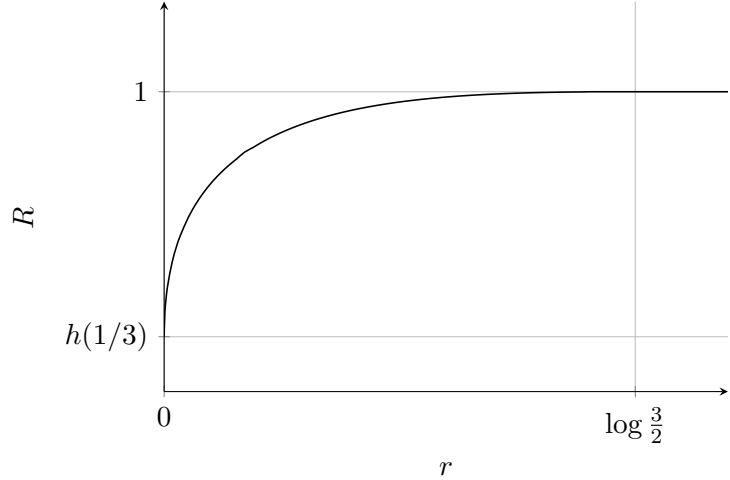
\begin{figure}
\begin{center}
\begin{tikzpicture}
\begin{axis}[
  rateplotstyle,
  ytick={0.666666666666,0.9182958340544896,1},
  yticklabels={$\frac{2}{3}$,$h(1/3)$,$1$},
  ymin=0.9,
  ymax=1.03,
  xtick={0,0.584962500721156},
  xticklabels={0,$\log\frac{3}{2}$},
  extra y ticks = {0.6},
]
\addplot[plotline,smooth] table [x=r,y=R] {rAndRW.dat} -- (1,1);
\end{axis}
\end{tikzpicture}
\end{center}
\caption{\label{fig:WtoEPRrates} The optimal rate $R$ as a function of the strong converse exponent $r$ for transforming \W{} states into \EPR{} pairs between Alice and Bob. The rate saturates at $r=\log\frac{3}{2}$.}
\end{figure}

\subsection{General states}

For general states, one might try to follow a similar route starting with the upper bound
\begin{equation}
\begin{split}
R^*(\psi\to\EPR_{AB},r)
 & \le \inf_{\substack{\alpha\in[0,1)  \\  x\in[0,1]}}\left[r\frac{\alpha}{1-\alpha}+E^{\alpha,(x,1-x,0)}(\psi)\right]  \\
 & = \inf_{\substack{\alpha\in[0,1)  \\  x\in[0,1]}}\sup_{\overline{\lambda}_A,\overline{\lambda}_B,\overline{\lambda}_C}\left[r\frac{\alpha}{1-\alpha}+x\entropy(\overline{\lambda}_A)+(1-x)\entropy(\overline{\lambda}_B)-\frac{\alpha}{1-\alpha}I_\psi(\overline{\lambda}_A,\overline{\lambda}_B,\overline{\lambda}_C)\right].
\end{split}
\end{equation}
If $I_\psi$ was convex, then an argument similar to the one in the preceding section would lead to an analogous formula with the optimization over $Q$ replaced with $(\overline{\lambda}_A,\overline{\lambda}_B,\overline{\lambda}_C$ subject to $I_\psi(\overline{\lambda}_A,\overline{\lambda}_B,\overline{\lambda}_C)\le r$, and $\min\{\entropy(\overline{\lambda}_A),\entropy(\overline{\lambda}_B)\}$ instead of the minimum of the two marginal entropies. However, $I_\psi$ is not known to be convex in general.

To overcome this difficulty, we introduce a modification of the rate function that contains less information, but which is convex, and for our purposes interchangeble with $I_\psi$. We introduce the construction and prove its basic properties in a more general context.
\begin{definition}
Let $(X,\preccurlyeq)$ be a compact pospace, i.e., $X$ is a compact topological space and ${\preccurlyeq}\subseteq X\times X$ a closed partial order. For an extended real-valued function $f:X\to\reals\cup\{\infty\}$, not necessarily continuous but bounded from below, we define
\begin{equation}
f^\searrow(y)=\inf_{\substack{x\in X  \\  x\preccurlyeq y}}f(x).
\end{equation}
\end{definition}

Recall that the domain of an extended real-valued function $f:X\to\reals\cup\{\infty\}$ is $\domain f=\setbuild{x\in X}{f(x)<\infty}$, and the characteristic function of a subset $A\subseteq X$ is
\begin{equation}
\chi_A(x)=\begin{cases}
0 & \text{if $x\in A$}  \\
\infty & \text{if $x\notin A$.}
\end{cases}
\end{equation}
$f$ is lower semicontinuous if its epigraph $\epigraph f=\setbuild{(x,z)\in X\times\reals}{f(x)\le z}$ is closed. $\chi_A$ is lower semicontinuous if and only if $A$ is closed.

\begin{proposition}\label{prop:fsearrowproperties}
Let $(X,\preccurlyeq)$ be a compact pospace and $f:X\to\reals\cup\{\infty\}$.
\begin{enumerate}
\item\label{it:fsearrowlsc} If $f$ is lower semicontinuous, then $f^\searrow$ is also lower semicontinuous.
\item\label{it:fsearrowdom} If $\domain f$ is compact, then $\domain f^\searrow$ is also compact.
\item\label{it:fsearrowdec} $f^\searrow$ is monotone decreasing.
\end{enumerate}
\end{proposition}
\begin{proof}
We use that if $X$ is a compact topological space and $Y$ is an arbitrary topological space, then the projection $p_2:X\times Y\to Y$ is a closed map.

\ref{it:fsearrowlsc}: $\hat{f}:=\chi_\preccurlyeq+(f\circ p_1):X\times X\to\reals\cup\{\infty\}$ is a lower semicontinuous function, and $\epigraph f^\searrow$ is the projection of $\epigraph \hat{f}$ along $p_2:X\times X \to X$, therefore it is closed.

\ref{it:fsearrowdom}: We may express the domain of $f^\searrow$ as
\begin{equation}
\begin{split}
\domain f^\searrow
 & = \setbuild{y\in X}{\exists x\in\domain f:x\preccurlyeq y}  \\
 & = p_2(\preccurlyeq\cap(\domain f\times X)),
\end{split}
\end{equation}
which is a continuous image of a compact set, therefore compact.

\ref{it:fsearrowdec}: If $y\preccurlyeq y'$, then $x\preccurlyeq y$ implies $x\preccurlyeq y'$, therefore
\begin{equation}
f^\searrow(y)
 = \inf_{\substack{x\in X  \\  x\preccurlyeq y}}f(x)
 \ge \inf_{\substack{x\in X  \\  x\preccurlyeq y'}}f(x)
 = f^\searrow(y').
\end{equation}
\end{proof}
In fact, $f^\searrow$ is the (pointwise) largest monotone decreasing function not exceeding $f$.

Returning to the rate function, our next goal is to prove that $\domain I_\psi^\searrow$ is compact and $I_\psi^{\smash{\searrow}}$ is lower semicontinuous and convex. We show the first two properties for $I_\psi$ and then \cref{prop:fsearrowproperties} implies that they also hold for $I_\psi^{\smash{\searrow}}$.
\begin{proposition}
Let $\psi\in\mathcal{H}_A\otimes\mathcal{H}_B\otimes\mathcal{H}_C$. Then $\domain I_\psi$ is compact and $I_\psi$ is lower semicontinuous.
\end{proposition}
\begin{proof}
That $\domain I_\psi$ is compact follows from the fact that it is equal to the entanglement polytope of $\psi$, which is a direct consequence of the invariant-theoretic characterisation of the entanglement polytope and the semigroup property of the highest weights of nonvanishing covariants \cite[Theorem 1]{walter2013entanglement}.

To see that $I_\psi$ is lower semicontinuous, let $\overline{\lambda}_0\in\overline{\partitions}^3$ and $\epsilon>0$, and let $\overline{\mu}\in\ball{\epsilon/2}{\overline{\lambda}_0}$ be such that $I_\psi(\overline{\mu})\le\liminf_{\overline{\lambda}\to\overline{\lambda}_0}I_\psi(\overline{\lambda})+\epsilon$. Then we have
\begin{equation}
\begin{split}
\lim_{n\to\infty}-\frac{1}{n}\log\sum_{\substack{\lambda\in\partitions[n]^3  \\  \frac{1}{\lambda}\in\ball{\epsilon}{\overline{\lambda}_0}}}\norm{P_\lambda\psi^{\otimes n}}^2
 & \le \lim_{n\to\infty}-\frac{1}{n}\log\sum_{\substack{\lambda\in\partitions[n]^3  \\  \frac{1}{\lambda}\in\ball{\epsilon/2}{\overline{\mu}}}}\norm{P_\lambda\psi^{\otimes n}}^2  \\
 & \le I_\psi(\overline{\mu})  \\
 & \le \liminf_{\overline{\lambda}\to\overline{\lambda}_0}I_\psi(\overline{\lambda})+\epsilon.
\end{split}
\end{equation}
Taking the limit $\epsilon\to 0$ gives
\begin{equation}
I_\psi(\overline{\lambda}_0)\le\liminf_{\overline{\lambda}\to\overline{\lambda}_0}I_\psi(\overline{\lambda}).
\end{equation}
\end{proof}

Before proving convexity, we need another expression for $I_\psi^\searrow$. For an arbitrary subset $U\subseteq\overline{\partitions}^3$ and all $m\in\naturals$ we introduce the projections
\begin{equation}
Q^m_U=\sum_{\substack{\mu\in\partitions[m]^3  \\  \exists\overline{\mu}\in U:\frac{1}{m}\mu\preccurlyeq\overline{\mu}}}P_\mu.
\end{equation}
\begin{proposition}
For $U,V\subseteq\overline{\partitions}^3$ and $m,n\in\naturals$ the inequality
\begin{equation}
Q^m_U\otimes Q^n_V\le Q^{m+n}_{\frac{mU+nV}{m+n}}
\end{equation}
holds.
\end{proposition}
\begin{proof}
For all $\mu\in\partitions[m]$, $\nu\in\partitions[n]$ and $\lambda\in\partitions[m+n]$ the projections $P_\lambda(P_\mu\otimes P_\nu)\in\boundeds(\mathcal{H}_A^{\otimes(m+n)}\otimes\mathcal{H}_B^{\otimes(m+n)}\otimes\mathcal{H}_C^{\otimes(m+n)})$ are orthogonal to each other, and nonzero if and only if the Littlewood--Richardson coefficient $c^\lambda_{\mu\nu}$ does not vanish, which implies $\lambda\preccurlyeq\mu+\nu$. Therefore
\begin{equation}
\begin{split}
Q^m_U\otimes Q^n_V
 & = \sum_{\substack{\mu\in\partitions[m]^3  \\  \exists\overline{\mu}\in U:\frac{1}{m}\mu\preccurlyeq\overline{\mu}}}\sum_{\substack{\nu\in\partitions[n]^3  \\  \exists\overline{\nu}\in V:\frac{1}{n}\nu\preccurlyeq\overline{\nu}}}P_\mu\otimes P_\nu  \\
 & = \sum_{\substack{\mu\in\partitions[m]^3  \\  \exists\overline{\mu}\in U:\frac{1}{m}\mu\preccurlyeq\overline{\mu}}}\sum_{\substack{\nu\in\partitions[n]^3  \\  \exists\overline{\nu}\in V:\frac{1}{n}\nu\preccurlyeq\overline{\nu}}}\sum_{\lambda\in\partitions[m+n]}P_\lambda(P_\mu\otimes P_\nu)  \\
 & = \sum_{\substack{\mu\in\partitions[m]^3  \\  \exists\overline{\mu}\in U:\frac{1}{m}\mu\preccurlyeq\overline{\mu}}}\sum_{\substack{\nu\in\partitions[n]^3  \\  \exists\overline{\nu}\in V:\frac{1}{n}\nu\preccurlyeq\overline{\nu}}}\sum_{\substack{\lambda\in\partitions[m+n]  \\  \lambda\preccurlyeq\mu+\nu}}P_\lambda(P_\mu\otimes P_\nu)  \\
 & = \sum_{\substack{\mu\in\partitions[m]^3  \\  \exists\overline{\mu}\in U:\frac{1}{m}\mu\preccurlyeq\overline{\mu}}}\sum_{\substack{\nu\in\partitions[n]^3  \\  \exists\overline{\nu}\in V:\frac{1}{n}\nu\preccurlyeq\overline{\nu}}}\sum_{\substack{\lambda\in\partitions[m+n]  \\  \exists\overline{\mu}\in U\,\exists\overline{\nu}\in V:\lambda\preccurlyeq m\overline{\mu}+n\overline{\nu}}}P_\lambda(P_\mu\otimes P_\nu)  \\
 & \le \sum_{\mu\in\partitions[m]^3}\sum_{\nu\in\partitions[n]^3}\sum_{\substack{\lambda\in\partitions[m+n]  \\  \exists\overline{\mu}\in U\,\exists\overline{\nu}\in V:\lambda\preccurlyeq m\overline{\mu}+n\overline{\nu}}}P_\lambda(P_\mu\otimes P_\nu)  \\
 & = \sum_{\substack{\lambda\in\partitions[m+n]  \\  \exists\overline{\mu}\in U\,\exists\overline{\nu}\in V:\lambda\preccurlyeq m\overline{\mu}+n\overline{\nu}}}P_\lambda  \\
 & = Q^{m+n}_{\frac{mU+nV}{m+n}}.
\end{split}
\end{equation}
\end{proof}
\begin{corollary}\label{cor:supermultiplicative}
If $\overline{\mu},\overline{\nu}\in\overline{\partitions}^3$, $m,n\in\naturals$ and $\epsilon>0$, then
\begin{equation}
Q^m_{\ball{\epsilon}{\overline{\mu}}}\otimes Q^n_{\ball{\epsilon}{\overline{\nu}}}\le Q^{m+n}_{\ball{\epsilon}{\frac{m}{m+n}\overline{\mu}+\frac{n}{m+n}\overline{\nu}}}.
\end{equation}
If in addition $\psi\in\mathcal{H}_A\otimes\mathcal{H}_B\otimes\mathcal{H}_C$, then
\begin{equation}
\norm{Q^m_{\ball{\epsilon}{\overline{\mu}}}\psi^{\otimes m}}^2\norm{Q^n_{\ball{\epsilon}{\overline{\nu}}}\psi^{\otimes n}}^2\le \norm{Q^{m+n}_{\ball{\epsilon}{\frac{m}{m+n}\overline{\mu}+\frac{n}{m+n}\overline{\nu}}}\psi^{\otimes m+n}}^2.
\end{equation}
\end{corollary}

\begin{proposition}
\begin{equation}
I_\psi^\searrow(\overline{\lambda})=\lim_{\epsilon\to 0}\lim_{n\to\infty}-\frac{1}{n}\log\norm{Q^n_{\ball{\epsilon}{\overline{\lambda}}}\psi^{\otimes n}}^2.
\end{equation}
\end{proposition}
\begin{proof}
The limit exists since by \cref{cor:supermultiplicative} the sequence $n\mapsto \log\norm{Q^n_{\ball{\epsilon}{\overline{\lambda}}}\psi^{\otimes n}}^2$ is superadditive, and monotone in $\epsilon$. Let $\overline{\lambda}\in\overline{\partitions}^3$. For all $m\in\positiveintegers$ choose $n_m\in\naturals$ such that the sequence $(n_m)_{m\in\positiveintegers}$ is strictly increasing and
\begin{equation}
\lim_{n\to\infty}-\frac{1}{n}\log\norm{Q^n_{\ball{1/m}{\overline{\lambda}}}\psi^{\otimes n}}^2 \ge -\frac{1}{n_m}\log\norm{Q^{n_m}_{\ball{1/m}{\overline{\lambda}}}\psi^{\otimes n_m}}^2-\frac{1}{m},
\end{equation}
and let $\mu_m\in\partitions[n_m]^3$ such that $\frac{1}{n_m}\mu_m$ is majorized by some element $\overline{\lambda}^{(m)}$ of $\ball{1/m}{\overline{\lambda}}$ and $\norm{P_{\mu_m}\psi^{n_m}}$ is maximal under this condition. Then
\begin{equation}
\norm{Q^{n_m}_{\ball{1/m}{\overline{\lambda}}}\psi^{\otimes n_m}}^2\le(n_m+1)^{d_A+d_B+d_C} \norm{P_{\mu_m}\psi^{n_m}}^2,
\end{equation}
therefore
\begin{equation}
\lim_{\epsilon\to 0}\lim_{n\to\infty}-\frac{1}{n}\log\norm{Q^n_{\ball{\epsilon}{\overline{\lambda}}}\psi^{\otimes n}}^2
 \ge \lim_{m\to\infty}-\frac{1}{n_m}\log\norm{P_{\mu_m}\psi^{n_m}}^2.
\end{equation}
Let $\mu_{m_l}$ be a subsequence such that $\frac{1}{n_{m_l}}\mu_{m_l}$ has a limit $\overline{\mu}$. Since $\frac{1}{n_{m_l}}\mu_{m_l}\preccurlyeq\overline{\lambda}^{(m_l)}$ for all $l$ and majorization is a closed partial order, $\overline{\mu}\preccurlyeq\lim_{l\to\infty}\overline{\lambda}^{(m_l)}=\overline{\lambda}$. It follows that
\begin{equation}
\begin{split}
I_\psi^\searrow(\overline{\lambda})
 & \le I_\psi(\overline{\mu})  \\  
 & \le\liminf_{l\to\infty}-\frac{1}{n_{m_l}}\log\norm{P_{\mu_{m_l}}\psi^{n_{m_l}}}^2  \\
 & \le \lim_{\epsilon\to 0}\lim_{n\to\infty}-\frac{1}{n}\log\norm{Q^n_{\ball{\epsilon}{\overline{\lambda}}}\psi^{\otimes n}}^2
\end{split}
\end{equation}

For the reverse inequality, suppose that $\overline{\mu}\preccurlyeq\overline{\lambda}$ and choose a sequence $\mu_n\in\partitions[n]^3$ such that $\frac{1}{n}\mu_n=\overline{\mu}$ and
\begin{equation}
\lim_{n\to\infty}-\frac{1}{n}\log\norm{P_{\mu_n}\psi^{\otimes n}}^2=I_\psi(\overline{\mu})
\end{equation}
For any $\epsilon>0$ and all sufficiently large $n$ there is an element of $\ball{\epsilon}{\overline{\lambda}}$ that majorizes $\frac{1}{n}\mu_n$, therefore $Q^n_{\ball{\epsilon}{\overline{\lambda}}}\ge P_{\mu_n}$, which implies
\begin{equation}
\lim_{n\to\infty}-\frac{1}{n}\log\norm{Q^n_{\ball{\epsilon}{\overline{\lambda}}}\psi^{\otimes n}}^2
 \le \lim_{n\to\infty}-\frac{1}{n}\log\norm{P_{\mu_n}\psi^{\otimes n}}^2
 = I_\psi(\overline{\mu}).
\end{equation}
Taking the limit $\epsilon\to 0$ and then the infimum over $\overline{\mu}\preccurlyeq\overline{\lambda}$ gives
\begin{equation}
\lim_{\epsilon\to 0}\lim_{n\to\infty}-\frac{1}{n}\log\norm{Q^n_{\ball{\epsilon}{\overline{\lambda}}}\psi^{\otimes n}}^2
\le I_\psi^\searrow(\overline{\lambda}).
\end{equation}
\end{proof}

\begin{proposition}
$I_\psi^\searrow$ is convex.
\end{proposition}
\begin{proof}
Let $\overline{\mu},\overline{\nu}\in\overline{\partitions}^3$ and $q\in(0,1)$. By \cref{cor:supermultiplicative}, for all $n$ (large enough so that $\lfloor qn\rfloor\neq 0$ and $\lfloor qn\rfloor\neq n$)
\begin{equation}
\begin{split}
& -\frac{\lfloor qn\rfloor}{n}\frac{1}{\lfloor qn\rfloor}\log\norm{Q^{\lfloor qn\rfloor}_{\ball{\epsilon}{\overline{\mu}}}\psi^{\otimes \lfloor qn\rfloor}}^2-\frac{n-\lfloor qn\rfloor}{n}\frac{1}{n-\lfloor qn\rfloor}\log\norm{Q^{n-\lfloor qn\rfloor}_{\ball{\epsilon}{\overline{\nu}}}\psi^{\otimes (n-\lfloor qn\rfloor)}}^2  \\
& \ge -\frac{1}{n}\log\norm{Q^n_{\ball{\epsilon}{\frac{\lfloor qn\rfloor}{n}\overline{\mu}+\frac{n-\lfloor qn\rfloor}{n}\overline{\nu}}}\psi^{\otimes n}}^2  \\
& \ge -\frac{1}{n}\log\norm{Q^n_{\ball{2\epsilon}{q\overline{\mu}+(1-q)\overline{\nu}}}\psi^{\otimes n}}^2.
\end{split}
\end{equation}
In the last inequality using that $\frac{\lfloor qn\rfloor}{n}\overline{\mu}+\frac{n-\lfloor qn\rfloor}{n}\overline{\nu}\to q\overline{\mu}+(1-q)\overline{\nu}$, therefore the distance is eventually less than $\epsilon$. We take the limit $n\to\infty$ and then $\epsilon\to 0$ to get
\begin{equation}
qI_\psi^\searrow(\overline{\mu})+(1-q)I_\psi^\searrow(\overline{\mu})\ge I_\psi^\searrow(q\overline{\mu}+(1-q)\overline{\nu}).
\end{equation}
\end{proof}

We have so far established that $I_\psi^\searrow$ is convex, lower semicontinuous and $\domain I_\psi^\searrow$ is compact. Since $\domain I_\psi^\searrow$ is a polytope, it is also continuous. Now we can prove the analogue of \cref{prop:freesupportupperbound} for general tripartite states.
\begin{proposition}\label{prop:generalupperbound}
Let $\psi\in\mathcal{H}_A\otimes\mathcal{H}_B\otimes\mathcal{H}_C$ be a state vector. Then
\begin{equation}
R^*(\psi\to\EPR_{AB},r)\le\max_{\substack{\overline{\lambda}\in\overline{\partitions}^3  \\  I_\psi^\searrow(\overline{\lambda})\le r}}\min\{\entropy(\overline{\lambda}_A),\entropy(\overline{\lambda}_B)\}.
\end{equation}
\end{proposition}
\begin{proof}
Similarly to \cref{prop:freesupportupperbound}, we start with the upper bound
\begin{equation}
\begin{split}
R^*(\psi\to\EPR_{AB},r)
 & \le \frac{r\frac{\alpha}{1-\alpha}+E^{\alpha,\theta}(\psi)}{E^{\alpha,\theta}(\EPR_{AB})}  \\
 & = \frac{r\frac{\alpha}{1-\alpha}+E^{\alpha,\theta}(\psi)}{\theta(A)+\theta(B)},
\end{split}
\end{equation}
and optimize over $\alpha\in[0,1)$ and $\theta=(x,1-x,0)$:
\begin{equation}
\begin{split}
R^*(\psi\to\EPR_{AB},r)
 & \le \inf_{\substack{\alpha\in[0,1)  \\  x\in[0,1]}}\left[r\frac{\alpha}{1-\alpha}+E^{\alpha,(x,1-x,0)}(\psi)\right]  \\ 
 & = \inf_{\substack{\alpha\in[0,1)  \\  x\in[0,1]}}\max_{\overline{\lambda}\in\overline{\partitions}^3}\left[r\frac{\alpha}{1-\alpha}+x\entropy(\overline{\lambda}_A)+(1-x)\entropy(\overline{\lambda}_B)-\frac{\alpha}{1-\alpha}I_\psi(\overline{\lambda})\right]  \\ 
 & \le \inf_{\substack{\alpha\in[0,1)  \\  x\in[0,1]}}\max_{\overline{\lambda}\in\overline{\partitions}^3}\left[r\frac{\alpha}{1-\alpha}+x\entropy(\overline{\lambda}_A)+(1-x)\entropy(\overline{\lambda}_B)-\frac{\alpha}{1-\alpha}I_\psi^\searrow(\overline{\lambda})\right]  \\ 
 & = \inf_{\alpha\in[0,1)}\max_{\overline{\lambda}\in\overline{\partitions}^3}\left[r\frac{\alpha}{1-\alpha}+\min\{\entropy(\overline{\lambda}_A),\entropy(\overline{\lambda}_B)\}-\frac{\alpha}{1-\alpha}I_\psi^\searrow(\overline{\lambda})\right]  \\ 
 & = \inf_{t\in[0,\infty)}\max_{\overline{\lambda}\in\overline{\partitions}^3}\left[\min\{\entropy(\overline{\lambda}_A),\entropy(\overline{\lambda}_B)\}+t\left(r-I_\psi^\searrow(\overline{\lambda})\right)\right]  \\ 
 & = \max_{\overline{\lambda}\in\overline{\partitions}^3}\inf_{t\in[0,\infty)}\left[\min\{\entropy(\overline{\lambda}_A),\entropy(\overline{\lambda}_B)\}+t\left(r-I_\psi^\searrow(\overline{\lambda})\right)\right]  \\ 
 & = \max_{\substack{\overline{\lambda}\in\overline{\partitions}^3  \\  I_\psi^\searrow(\overline{\lambda})\le r}}\min\{\entropy(\overline{\lambda}_A),\entropy(\overline{\lambda}_B)\}.
\end{split}
\end{equation}
The first equality uses the definition of $E^{\alpha,\theta}$, the second inequality holds since $I_\psi\ge I_\psi^\searrow$, the second equality follows from interchanging the minimum over $x$ and the maximum over $\overline{\lambda}$ (the function is affine in $x$ and concave and upper semicontinuous in $\overline{\lambda}$), in the next step the new variable $t=\frac{\alpha}{1-\alpha}$ is introduced, the fourth equality uses Sion's minimax theorem \cite{sion1958general} (affine in $t$ and concave and upper semicontinuous in $\overline{\lambda}$), and the last step uses that the infimum over $t$ is $\min\{\entropy(\overline{\lambda}_A),\entropy(\overline{\lambda}_B)\}$ if the coefficient of $t$ is nonnegative and $-\infty$ otherwise.
\end{proof}

\begin{proposition}\label{prop:partitiontoratebound}
Let $n\in\positiveintegers$, $\lambda\in\partitions[n]^3$ such that $-\frac{1}{n}\log\norm{Q_\lambda\psi^{\otimes n}}^2\le r$. Then
\begin{equation}
R^*(\psi\to\EPR_{AB},r)\ge\min\left\{\entropy\left(\frac{1}{n}\lambda_A\right),\entropy\left(\frac{1}{n}\lambda_B\right)\right\}.
\end{equation}
\end{proposition}
\begin{proof}
By \cref{cor:supermultiplicative}, for all $m\in\positiveintegers$ we have $-\frac{1}{mn}\log\norm{Q_{m\lambda}\psi^{\otimes mn}}^2\le r$. With respect to the Schur--Weyl decomposition, the marginal of the normalization of $Q_{m\lambda}\psi^{\otimes mn}$ on subsystem $A$ has the form
\begin{equation}
\bigoplus_{\mu\preccurlyeq m\lambda_A}p_\mu\rho_\mu\otimes \frac{I_{[\mu]}}{\dim[\mu]}
\end{equation}
with $\rho_\mu\in\states(\mathbb{S}_\mu(\mathcal{H}_A))$, therefore its min-entropy is at least
\begin{equation}
\begin{split}
\min_{\mu\preccurlyeq m\lambda_A}\log\dim[\mu]
 & \ge \log\min_{\mu\preccurlyeq m\lambda_A}\frac{1}{(mn+d_A)^{(d_A+2)(d_A-1)/2}}2^{mn\entropy(\mu/mn)}  \\
 & \ge mn\entropy(\lambda_A/n)-\frac{(d_A+2)(d_A-1)}{2}\log(mn+d_A).
\end{split}
\end{equation}
Similarly, the min-entropy of the $B$ marginal is at least
\begin{equation}
mn\entropy(\lambda_B/n)-\frac{(d_B+2)(d_B-1)}{2}\log(mn+d_B)
\end{equation}

By \cref{thm:oneshotdeterministic}, the normalized state can be transformed into
\begin{equation}
mn\min\{\entropy(\lambda_A/n),\entropy(\lambda_B/n)\}-O(\log mn)
\end{equation}
copies of $EPR_{AB}$ deterministically, therefore
\begin{equation}
\begin{split}
R^*(\psi\to\EPR_{AB},r)
 & \ge\lim_{m\to\infty}\frac{1}{mn}\left(mn\min\{\entropy(\lambda_A/n),\entropy(\lambda_B/n)\}-O(\log mn)\right)  \\
 & = \min\left\{\entropy\left(\frac{1}{n}\lambda_A\right),\entropy\left(\frac{1}{n}\lambda_B\right)\right\}.
\end{split}
\end{equation}
\end{proof}

\begin{proposition}\label{prop:generallowerbound}
Let $\overline{\lambda}\in\overline{\partitions}^3$ such that $I_\psi^\searrow(\overline{\lambda})\le r$. Then
\begin{equation}
R^*(\psi\to\EPR_{AB},r)\ge\min\{\entropy(\overline{\lambda}_A),\entropy(\overline{\lambda}_B)\}.
\end{equation}
\end{proposition}
\begin{proof}
Choose a sequence $\epsilon_n\to 0$ such that
\begin{equation}
I_\psi^\searrow(\overline{\lambda})=\lim_{n\to\infty}-\frac{1}{n}\log\norm{Q^n_{\ball{\epsilon_n}{\overline{\lambda}}}\psi^{\otimes n}}^2.
\end{equation}
For each $n$, choose $\mu_n\in\partitions[n]^3$ such that $\mu_n$ is majorized by an element of $\ball{\epsilon_n}{\overline{\lambda}}$ and $\norm{P_{\mu_n}\psi^{\otimes n}}^2$ is maximal subject to this condition. Since there are polynomially many partitions, we also have
\begin{equation}
I_\psi^\searrow(\overline{\lambda})=\lim_{n\to\infty}-\frac{1}{n}\log\norm{P_{\mu_n}\psi^{\otimes n}}^2.
\end{equation}
Let $r_n=-\frac{1}{n}\log\norm{P_{\mu_n}\psi^{\otimes n}}^2$. By \cref{prop:partitiontoratebound}, for each $n$ we have
\begin{equation}
R^*(\psi\to\EPR_{AB},r_n)\ge\min\left\{\entropy\left(\frac{1}{n}\mu_{n,A}\right),\entropy\left(\frac{1}{n}\mu_{n,B}\right)\right\}.
\end{equation}
By compactness of $\overline{\partitions}^3$, along a subsequence $\frac{1}{n}\mu_n$ converges to some limit $\overline{\mu}$. Since majorization is a closed relation, $\overline{\mu}\preccurlyeq\overline{\lambda}$. Taking the limit along this subsequence, we therefore obtain
\begin{equation}
R^*(\psi\to\EPR_{AB},r)
\ge\min\{\entropy(\overline{\mu}_A),\entropy(\overline{\mu}_B)\}
\ge\min\{\entropy(\overline{\lambda}_A),\entropy(\overline{\lambda}_B)\}.
\end{equation}
\end{proof}

To summarize, \cref{prop:generalupperbound,prop:generallowerbound} together prove our general result:
\begin{theorem}\label{thm:generalscrate}
Let $\psi\in\mathcal{H}_A\otimes\mathcal{H}_B\otimes\mathcal{H}_C$ be a state vector and $r\ge 0$. Then
\begin{equation}
R^*(\psi\to\EPR_{AB},r)=\max_{\substack{\overline{\lambda}\in\overline{\partitions}^3  \\  I_\psi^\searrow(\overline{\lambda})\le r}}\min\{\entropy(\overline{\lambda}_A),\entropy(\overline{\lambda}_B)\}.
\end{equation}
\end{theorem}

\section{Conclusion}\label{sec:conclusion}

We determined the optimal rates for transforming many copies of a tripartite pure state into ebits between the first two of the subsystems, under various success criteria. Specifically, when either the success probability or the fidelity is required to be $1$ for any finite number of copies, or approach $1$ exponentially fast, the rate is given by the minimum of the rates for the two bipartite transformations. In contrast, when the success probability or the fidelity is allowed to approach $0$ with a specified strong converse exponent, the bipartite upper bound is not achievable, and the optimal rate is given by optimization over a set of R\'enyi-type tripartite entanglement measures. The rate formula in the strong converse case interpolates between two previously understood limits: with asymptotically vanishing error (either in a probabilistic or approximate sense), the optimal rate is the minimum of the two marginal von Neumann entropies \cite{smolin2005entanglement}, while for asymptotic SLOCC transformations the rate is the minimum over a one-parameter subfamily of the quantum functionals \cite{christandl2023universal,christandl2023weighted}. When the full range of strong converse exponents is considered, an additional order parameter appears, and the relevant entanglement measures can be seen as R\'enyi generalizations of the convex combinations of the two marginal von Neumann entropies, which are nevertheless not simply convex combinations themselves when the order is less than $1$ \cite{vrana2023family}.

An open problem suggested by our work is to find similar rate formulas for multipartite-to-bipartite transformations. Such a generalization is far from straightforward. In the vanishing-error case, the optimal rate is the minimum of the entanglement entropies for all cuts separating the $A$ and $B$ subsystems \cite{smolin2005entanglement,horodecki2005partial}. However, the optimal rate under asymptotic SLOCC is not known. A major obstacle to finding it as well as the optimal rate in the strong converse domain is that it likely involves R\'enyi generalizations of all convex combinations of the entanglement entropies for these cuts. While such generalizations are known to exist, our understanding of them is still very much limited \cite{bugar2024interpolating}. It seems reasonable to expect that for each convex combination and every order parameter $\alpha\in[0,1]$ there is a unique such R\'enyi type entanglement measure, and finding them explicitly might also bring the problem of multipartite-to-bipartite transformations within reach.

The direct regime looks somewhat simpler because at least in the tripartite-to-bipartite setting only involves bipartite entanglement measures. In fact, in the case of deterministic transformations it is not difficult to extend our results to get the exact optimal rate. The idea is to measure one subsystem at a time, using \cref{prop:goodvector} multiple times for each projection, with respect to different tripartitions. For instance, for three subsystems $ABCD$, while measuring on $D$ one needs to consider the tripartitions $A:BC:D$ and $AC:B:D$ to be able to conclude that the entanglement between $A$ and $BC$ remains close to the smallest of the min-entropies of entanglement for $A:BCD$ and $AD:BC$ and at the same time the entanglement between $B$ and $AC$ remains close to the smallest of the min-entropies of entanglement for $B:ACD$ and $BD:AC$. By looking at the proof of \cref{thm:oneshotdeterministic} one can see that the bounds still work when a (bounded) number of such tripartitions is simultaneously considered (by the union bound, one needs to add a bounded number of doubly exponentially small probabilities). In this way one can see that the maximal deterministic rate is equal to the minimum of the min-entropies of entanglement for all bipartitions separating $A$ and $B$. On the other hand, the case of exponential convergence of the success probability to $1$ is less clear. The difficulty is that \cref{prop:simultaneoustruncation} cannot be applied, simply because a truncation for e.g. the subsystem $BC$ is not a local operation.

Our work also implies an improved approximation of the entanglement measures $E^{\alpha,\theta}$. The methods of \cref{sec:strongconverse} (when extended to multipartite systems) and the upper bound of \cite[Proposition 3.6]{bugar2024explicit} gives that if $\psi\in\mathcal{H}_1\otimes\dots\otimes\mathcal{H}_k$ with local dimensions $d_1,\dots,d_k$ and $n\in\positiveintegers$, then for
\begin{equation}
M_n:=\max_{\lambda\in\partitions[n]^k}\left[\sum_{j=1}^k\theta(j)\entropy\left(\frac{1}{n}\lambda_j\right)+\frac{\alpha}{1-\alpha}\frac{1}{n}\log\norm{P_\lambda\psi^{\otimes n}}\right]
\end{equation}
the inequality
\begin{equation}
M_n\le E^{\alpha,\theta}(\psi) \le M_n + \left[\sum_{j=1}^k\theta(j)d_j+\frac{\alpha}{1-\alpha}\sum_{j=1}^kd_j\right]\frac{\log(n+1)}{n}
\end{equation}
holds. An undesirable feature of this approximation guarantee is its dependence on $\alpha$, more precisely that the difference between the lower and upper bounds approaches infinity as $\alpha\to 1$. Since $\lim_{\alpha\to 1}E^{\alpha,\theta}(\psi)=\sum_{j=1}^k\theta(j)\entropy(j)_\psi$, we expect that a better approximation is possible near $\alpha=1$, and in the tripartite case this would facilitate the computation of $R^*(\psi\to\EPR_{AB},r)$ for general states.

\section{Acknowledgement}

This work was supported by the Ministry of Culture and Innovation of Hungary from the National Research, Development and Innovation Fund via the research grants FK~146643, K~146380, and EXCELLENCE~151342, by the János Bolyai Research Scholarship of the Hungarian Academy of Sciences, and by the Ministry of Culture and Innovation and the National Research, Development and Innovation Office within the Quantum Information National Laboratory of Hungary (Grant No.~2022-2.1.1-NL-2022-00004).

\bibliography{refs}{}

\end{document}